%% file: MAIN.tex
\documentclass[10pt,journal,compsoc]{IEEEtran}

\usepackage{cite}
\usepackage{amsmath}
\usepackage{amsthm}
\usepackage{amssymb}
\usepackage{graphicx} 
\usepackage{subfigure}
\usepackage{url}
\usepackage{booktabs} 
\usepackage{array} 
\usepackage{verbatim} 
\usepackage{caption}
\usepackage{enumitem}
\newtheorem{mylemma}{Lemma}

\usepackage{multirow}
\usepackage{color}
\newcommand{\blue}[1]{\textcolor[rgb]{0.00,0.00,0.00}{{#1}}}
\newcommand{\mjrcolor}[1]{\textcolor[rgb]{0.00,0.00,0.00}{{#1}}}

\newcommand{\eq}[1]{Eq.~\eqref{#1}}
\usepackage{soul}
\setstcolor{red}
\newcommand{\myitem}[1]{\vspace{0.25\baselineskip}\noindent\textbf{#1}}

\definecolor{orange}{rgb}{1,0.5,0}

\usepackage{xspace}
\newcommand{\ourAlgo}[0]{\texttt{CABaRet}\xspace}
\newcommand{\secref}[1]{Section~\ref{#1}}

\newcommand{\mjr}[1]{\mjrcolor{#1}}
\newcommand{\mjrout}[1]{}

\usepackage{algorithm}
\usepackage{algorithmicx}
\usepackage{algpseudocode}

\begin{document}

\title{Network-aware Recommendations in the Wild: Methodology, Realistic Evaluations, Experiments}


\author{Savvas~Kastanakis\textsuperscript{1,2}, 
		Pavlos~Sermpezis\textsuperscript{3}, 
		Vasileios~Kotronis\textsuperscript{2},\\
		Daniel  Menasch\'e\textsuperscript{4},
		Thrasyvoulos Spyropoulos\textsuperscript{5}\\
		~\\
        \textsuperscript{1}~University of Crete, Greece; 
        \textsuperscript{2}~FORTH-ICS,Greece; 
        \textsuperscript{3}~Aristotle University of Thessaloniki, Greece; 
        \textsuperscript{4}~Federal University of Rio de Janeiro, Brazil; 
        \textsuperscript{5}~EURECOM, Sophia-Antipolis, France
        }

\maketitle

\begin{abstract}
\input{abstract}
\end{abstract}

\section{Introduction}
\label{sec:intro}
\input{Introduction}

\section{Methodology}
\label{sec:methodology}
\input{RecAlgo}


\section{Measurements and Evaluation}
\label{sec:measurements}
\input{Measurements}

\subsection{\blue{Results: ``Search Bar'' Video Demand}}
\label{sec:demand-google}
\input{Demand_google_trends}


\section{\blue{Experiments with Real Users}}
\label{sec:experiments}
\subsection{Experimental Testbed}
\input{ExperimentSetup}
\subsection{Results}
\input{Experiments}

\section{Related Work}
\label{sec:related}
\input{Related}

\section{Conclusion}
\label{sec:conclusions}
\input{Conclusions}

\section*{Acknowledgements} 
This research is co-financed by Greece and the European Union (European Social Fund- ESF) through the Operational Programme ``Human Resources Development, Education and Lifelong Learning'' in the context of the project ``Reinforcement of Postdoctoral Researchers - 2nd Cycle'' (MIS-5033021), implemented by the State Scholarships Foundation (IKY).





\bibliographystyle{IEEEtran}
\bibliography{Cacherec}
\input{biographies}
\vfill

\newpage
\appendices
\input{JointCacheRec}
\subsection{Results under Greedy Caching}\label{sec:results-greedy}
\input{measurements_greedy}

\end{document}

%% file: abstract.tex
\blue{Joint caching and recommendation has been recently proposed as a new paradigm for increasing the efficiency of mobile edge caching. Early findings demonstrate significant gains for the network performance. However, previous works evaluated the proposed schemes exclusively on simulation environments. Hence, it still remains uncertain whether the claimed benefits would change in real settings. In this paper, we propose a methodology that enables to evaluate joint network and recommendation schemes in real content services by only using publicly available information. We apply our methodology to the YouTube service, and conduct extensive measurements to investigate the potential performance gains. Our results show that significant gains can be achieved in practice; e.g., 8 to 10 times increase in the cache hit ratio from cache-aware recommendations. Finally, we build an experimental testbed and conduct experiments with real users; we make available our code and datasets to facilitate further research. To our best knowledge, this is the first realistic evaluation (over a real service, with real measurements and user experiments) of the joint caching and recommendations paradigm. Our findings provide experimental evidence for the feasibility and benefits of this paradigm, validate assumptions of previous works, and provide insights that can drive future research.}


%% file: Introduction.tex
Multi-access Edge Computing (MEC) is identified as one of the key technologies for 5G networks~\cite{MEC:white}. MEC architectures enable the extension of the successful paradigm of Content Delivery Networks (CDNs) and content caching to the edge of the mobile networks, thus reducing latency of content delivery and offloading of the backhaul links. However, a key difference to CDNs is that caches in MEC are located at the edge of the mobile network (e.g., base stations), and unavoidably have limited capacity and serve small --and frequently changing-- user populations~\cite{Paschos-infocom2016}. These factors, despite the advances in caching policies~\cite{Paschos-infocom2016} or delivery techniques~\cite{femto}, limit the possible gains from MEC: capacity is a tiny fraction of today's content catalogs, and traffic is highly variable; hence, a large number of user requests is for non-cached contents, i.e., not served in the edge.

A recently proposed solution for increasing the efficiency of MEC is jointly considering caching and recommending content~\cite{sch-chants-2016,chatzieleftheriou2017caching,sermpezis2018soft,giannakas-wowmom-2018,zhu2018coded,chatzieleftheriou2019jointly,costantini2019approximation,garetto2020similarity,tsigkari2020user,li2020leveraging}. Recommendation Systems (RS) are integrated in many popular services (e.g., YouTube, Netflix) and significantly affect the user demand~\cite{RecImpact-IMC10, gomez2016netflix}. \blue{Therefore, leveraging recommendations to steer content demand towards cached contents can significantly increase the cache hit ratio (network performance) and content delivery latency (user experience), even under the challenging conditions of small caches or populations in MEC.}



\blue{As a toy example of the proposed paradigm, consider the following (depicted in Fig.~\ref{fig:na-rec-example}): Assume a user watching a video A over a streaming service, whose recommendation system would suggest to the user to watch next a video B. Also assume that video B is not locally cached or needs to be fetched from a congested link, which does not allow a high quality streaming of B (low video quality, high start up delay, rebufferings, etc.). In the proposed paradigm, the recommendation system could instead suggest to the user to watch next a video C, which is still relevant to A (e.g., C is similar with B) and can be delivered in high quality (e.g., it is stored in a MEC cache). This network-aware recommendation for video C can be a win-win situation for the network, which consumes less resources for the video delivery, and for the user, who enjoys a better streaming experience.}

\blue{Previous works generalize the above example~\cite{sch-chants-2016,chatzieleftheriou2017caching,sermpezis2018soft,giannakas-wowmom-2018,zhu2018coded,lin2018joint,song2018making,qi2018optimizing,chatzieleftheriou2019jointly,giannakas2019order,chatzieleftheriou2019joint,gupta2019effect,lin2019content,costantini2019approximation,garetto2020similarity,tsigkari2020user,li2020leveraging}, by considering
more general recommendation techniques and parameters (e.g., number and order of recommendations~\cite{giannakas2019order}), more general content delivery schemes (e.g., multiple caches~\cite{sermpezis2018soft,chatzieleftheriou2019joint}, broadcasting~\cite{lin2018joint,song2018making,lin2019content}), and more general user demand models (e.g., acceptance of recommendations~\cite{qi2018optimizing}, sequential requests~\cite{giannakas-wowmom-2018}). 
}

\begin{figure}
\centering
\includegraphics[width=1\linewidth]{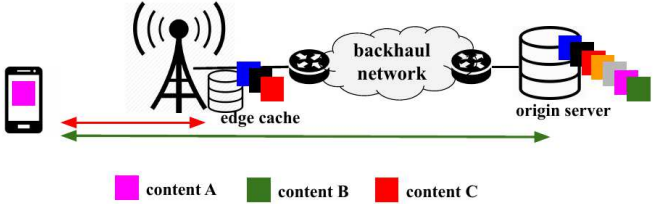}
\caption{Example of network-aware recommendations: Contents B and C are relevant to a content A currently consumed by a user. A baseline RS would recommend content B, while a network-aware RS recommends content C that can be served by the edge cache in a higher delivery quality.}
\label{fig:na-rec-example}
\end{figure}

\blue{Early findings demonstrate that the potential gains for the network performance can be significant, e.g., by increasing up to an order of magnitude the caching efficiency~\cite{sermpezis2018soft}. However, these promising results are based exclusively on evaluations on \textit{simulation environments} and mainly consider \textit{small content catalogs} (typically, a few thousands contents) of \textit{synthetic or public datasets that are not collected from real content delivery services} (e.g., MovieLens~\cite{movielens-related-dataset}). 
While the contribution of previous works to the understanding of the involved challenges, benefits and tradeoffs, is indisputable, it still remains uncertain \textit{if and how these findings would change in real settings}.} 

\blue{Deviations from such real setting aspects may affect the expected performance. For instance: (i) Real content delivery services typically have huge content catalogs (e.g., YouTube and Netflix video catalogs are reported to be in the order of petabytes), from which recommendations are selected and within which user are allowed to navigate. Considering only a tiny subset of these options for user actions may overestimate the gains. (ii) While the employed user models take into account the quality of recommendations and willingness of users to follow a ``nudged'' list of recommendations, we still lack experimental evidence whether this would indeed approximate well the real user behavior.}

\blue{In this paper, we aim to address these issues and take the next step in the evaluation of the joint network and recommendations paradigm: we propose a methodology that enables evaluation under realistic settings, apply it in a real service (YouTube), and conduct measurements and experiments with real users; to our best knowledge this is the first study of this kind in the field\footnote{\blue{A preliminary version of our work, containing some parts of this paper, appears in~\cite{kastanakis-cabaret-mecomm-2018}.}}. Specifically, our contributions are summarized as follows.}

\myitem{Methodology.} \blue{We propose a methodology that enables realistic evaluations for the joint network and recommendations paradigm (\secref{sec:methodology}). We exploit information made publicly available from the recommendation systems of content providers, and incorporate it to the existing frameworks used for the evaluation of joint network and recommendation schemes.}
\blue{In this way, we circumvent the problem of the content similarity data that is required by the majority of existing works, but is not disclosed by the content providers. In fact, we claim that detailed content/user information is not necessary, and only the output of a RS may suffice for the design of joint policies. This allows to face RSs as black boxes (without disclosing sensitive/private data and algorithms), and as a side-effect, it could enable joint caching and recommendation approaches, without requiring tight collaboration between network operators and CPs (as is considered by previous works).}

\blue{We apply the proposed approach, and design a network-aware algorithm (named \ourAlgo) that leverages available information provided by a RS, and returns cache-aware recommendations (\secref{sec:cabaret}).}

\myitem{Measurements and a realistic evaluation.} We apply our methodology and perform extensive measurements and evaluation over the YouTube service (\secref{sec:measurements}). Our results show that significant caching gains can be achieved in practice; even in conservative scenarios \blue{of the considered setup}, our approach increases the cache hit ratio by a factor of $\times$8 to $\times$10. %
%
\blue{To our best knowledge, this is the first evaluation of a joint caching and recommendation approach in a real service and with realistic traffic.}

\myitem{Experiments with real users.} \blue{We build an experimental testbed and conduct experiments with \textit{real users} to (i) test the performance of \ourAlgo in practice, and (ii) verify to what extent the assumptions made hold in practice with real users (\secref{sec:experiments}). We are the first to test the concept of joint caching and recommendations with real user experiments. Moreover, we publish the collected dataset and open-source the code of the testbed, which is generic and enables researchers to design and conduct their own experiments either with \ourAlgo or with any other algorithm they implement.}

\blue{The experimental findings (i) are in agreement with the measured performance of \ourAlgo in~\secref{sec:measurements}, which further supports the usefulness of the proposed methodology (\secref{sec:methodology}) for realistic evaluations; (ii) validate key assumptions made in related literature, and provide useful quantitative results that can drive future models, parameter selection and assumptions; and (iii) provide valuable evidence for the feasibility and benefits of the network-aware recommendations in practice, namely, users are willing to follow ``nudged'' recommendations and they do not perceive this as a significant compromise in recommendation quality.}

\blue{Finally, we provide an overview of related work (\secref{sec:related}) and conclude our paper (\secref{sec:conclusions}).}

%% file: RecAlgo.tex
\subsection{Overview} \label{sec:methodology-overview}

\subsubsection{Motivation: the need for realistic evaluations.} 

\mjr{We consider a communication system where a set of contents can be delivered with lower cost for the network and/or in higher QoS. To simplify our discussion, in the remainder we refer to a caching system as in the example of Fig.~\ref{fig:na-rec-example}, where the set of cached contents can be delivered in low-cost/high-QoS. However, our methodology applies to generic communication setups, by simply assigning a cost value to each content, where the cost of delivering a content may depend on the network, the wireless channel conditions, the transmission (coded, broadcast, unicast), etc.}

\blue{The main approach in literature for the joint design of network and recommendations, assumes that a system (for recommendations and/or caching) has full knowledge of the content similarities or user preferences. Under this assumption, it considers that the ``baseline RS'' (see Fig.~\ref{fig:RS-BS}) would recommend to the user the contents with the highest similarity/relevance, while the ``network-aware RS'' (see Fig.~\ref{fig:RS-NA}) could also recommend contents with lower similarity, as soon as the similarity is above a threshold (e.g., quality of recommendations~\cite{giannakas-wowmom-2018}, user preference window~\cite{chatzieleftheriou2017caching}). This relaxation in the recommendation quality is the key behind the joint network and recommendations paradigm, since it allows to recommend contents of lower similarity/relevance but which can be delivered more efficiently and/or in higher QoS by the network.}

\begin{figure}
\centering
\subfigure[Baseline RS]{\includegraphics[width=1\columnwidth]{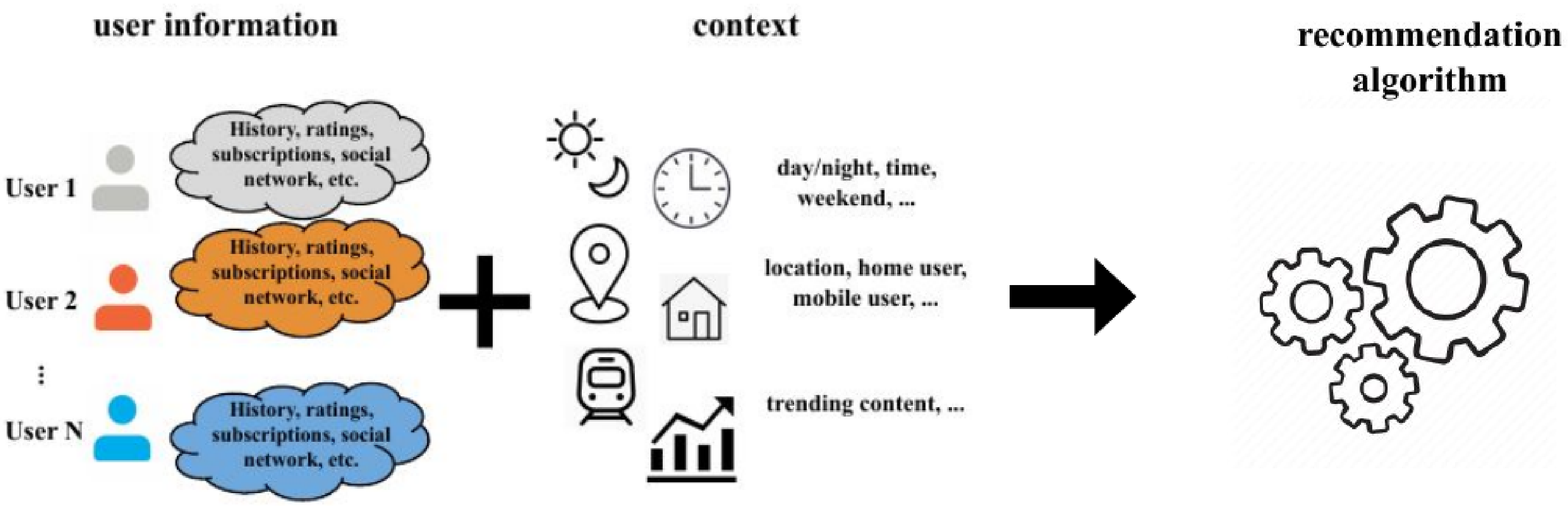}\label{fig:RS-BS}}

\subfigure[Network-aware RS (previous approaches)]{\includegraphics[width=1\columnwidth]{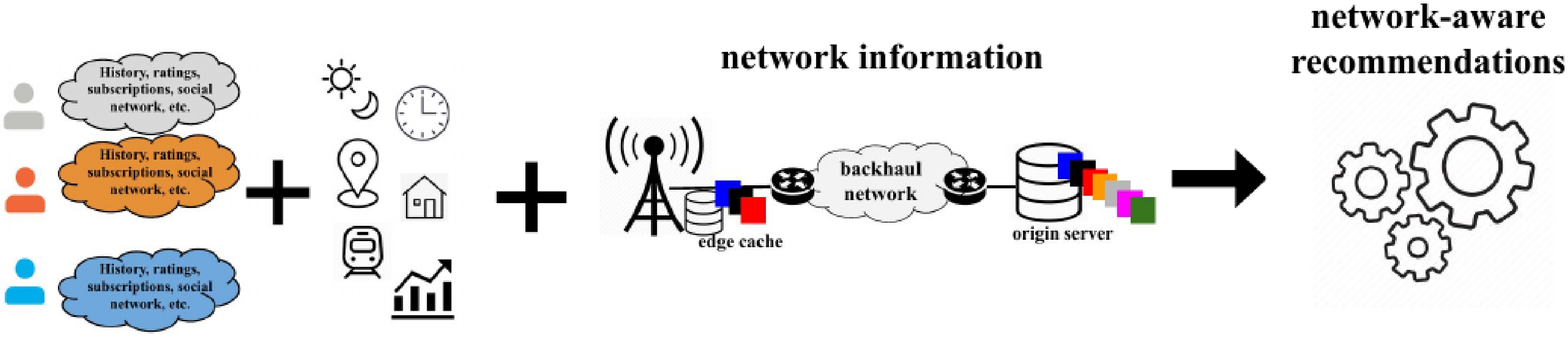}\label{fig:RS-NA}}

\subfigure[Network-aware RS (our approach)]{\includegraphics[width=1\columnwidth]{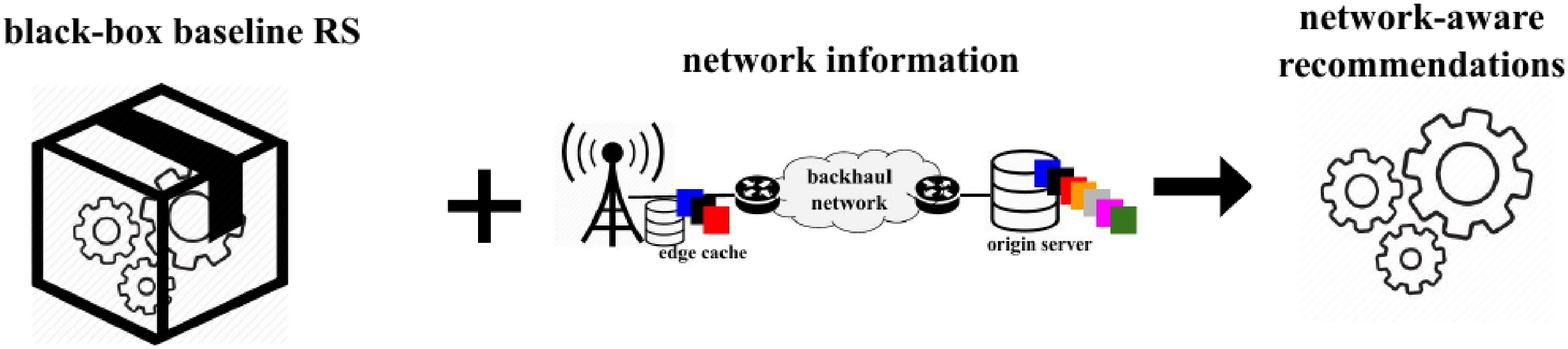}\label{fig:RS-NA-ours}}
\caption{Information that is taken into account by RS: (a) \textit{baseline RS}: detailed user information and context; (b) \textit{network-aware RS of previous approaches}: detailed user information (and context possibly) and network information; (c) \textit{the proposed approach for network-aware RS}: output of the baseline RS and network information.
}
\label{fig:RS-approaches}
\end{figure}

\blue{The full knowledge of content similarities/relevance can be a reasonable assumption when the network-aware RS is operated by the content provider, which holds all the information about contents and user activity. However, this approach comes with the following limitations.}

\blue{The information about content similarity and user preferences is not made public by content providers, due to its sensitive nature and business value. Hence, when it comes to the evaluation, previous works rely on synthetically generated datasets of content similarity or user preferences. It is unknown whether these datasets represent well real content catalogs, content similarity, user preferences, etc. In case they significantly deviate from the structure and characteristics of the real content similarity matrices (which are shown to be determinant factors for the performance~\cite{sch-chants-2016,sermpezis2018soft}), the evaluation may lead to erroneous or inaccurate findings and insights about the proposed solutions.}

\blue{Some works (e.g.,~\cite{sermpezis2018soft,giannakas-wowmom-2018,qi2018optimizing,chatzieleftheriou2019jointly,costantini2019approximation}) have considered publicly available data, such as the MovieLens dataset~\cite{movielens-related-dataset} that contains real user ratings for movies, and allow to calculate content similarities (e.g., using collaborative filtering techniques~\cite{chatzieleftheriou2019jointly}). However, this approach also deviates from real setups, as: (i) the sparsity of available data can lead to very different content similarities, (ii) the datasets do not correspond to content catalogs of real services and/or neglect the fact that content similarity/relevance is affected by time (trending content, transient interest in contents, etc.), and (iii) modern RS do not make recommendations only based on content similarity (collaborative filtering), but employ more complex mechanisms that take into account also user subscriptions, content diversity, trends, etc.~\cite{covington2016deep}.}

\subsubsection{The proposed methodology}


\blue{To overcome the aforementioned limitation, we propose a methodology that leverages information made publicly available by the RS of content services, and enables realistic evaluations on the real content catalogs and operation of these services. Instead of using the ``raw'' information about content similarity or user preferences, i.e., the \textit{input} of the baseline RS which is typically not publicly available, our methodology uses the \textit{output} of the baseline RS that is already provided by the content service. Figure~\ref{fig:RS-approaches} presents the conceptual difference between our approach (Fig.~\ref{fig:RS-NA-ours}) and previous network-aware approaches  (Fig.~\ref{fig:RS-NA}).}

\begin{figure}
\centering
\includegraphics[width=1\columnwidth]{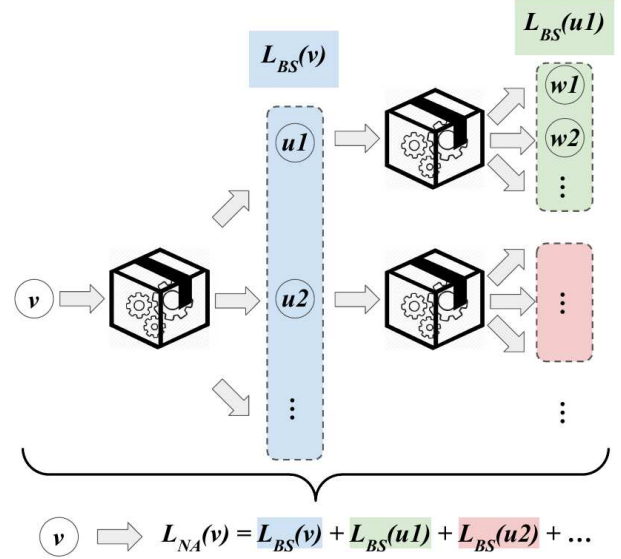}
\caption{Obtaining a list of ``relaxed'' recommendations by triggering the baseline RS in a breadth-first-search way.}
\label{fig:RS-cabaret-lists}
\end{figure}

\blue{In particular, our methodology is detailed below and depicted in Fig.~\ref{fig:RS-cabaret-lists}. Consider a content service (video streaming, online radio, etc.) that uses a baseline RS. When a user consumes a content $v$, the baseline RS provides a list $\mathcal{L}_{BS}(v)$ of contents related to $v$ and to the current session of the user (the subscript ``BS'' denotes the baseline RS). The recommendations in $\mathcal{L}_{BS}(v)$ are the most relevant according to the (network-oblivious) baseline RS. A network-aware RS (denoted with subscript ``NA'') would need a larger set of related contents $\mathcal{L}_{NA}(v)\supseteq \mathcal{L}_{BS}(v)$, including also less relevant contents, in order to be able to exploit a relaxation in the recommendation quality for the favor of network-friendly recommendations. We construct such a set of ``relaxed'' recommendations $\mathcal{L}_{NA}(v)$, by recursively triggering the baseline RS in a Breadth-First Search (BFS) manner: we start by obtaining the list $\mathcal{L}_{BS}(v)$, then $\forall  u\in\mathcal{L}_{BS}(v)$ we request from the baseline RS the lists  $\mathcal{L}_{BS}(u)$ and append them to $\mathcal{L}_{NA}(v)$, and so on. In the end of the process, the list $\mathcal{L}_{NA}(v)$ contains contents directly and indirectly related to $v$, which is the key information for joint network and recommendation algorithms.}

\blue{Hence, compared to previous approaches, the main advances of the proposed method are: (i) it can be applied to a real service, without being limited to samples or fractions of the content catalog; (ii) it leverages the real output of the baseline RS, thus avoiding simplifications for the baseline recommendation mechanisms  (e.g., naive item-item similarity, or collaborative filtering); (iii) it is based only on publicly available information and considers the baseline RS as a black-box, which enables to apply the approach on any real service, without requiring access to private/sensitive data.}

\subsubsection{Benefits beyond realistic evaluations}

\blue{Apart from enabling realistic evaluations, we believe that the proposed approach brings additional benefits for the joint network and recommendations paradigm, which we briefly discuss below.}

\myitem{Modular architectures.} \blue{The black-box approach yields a modular design, minimizing the dependence between network and RS components, e.g., the RS can be replaced or upgraded while the remainder of the system remains unchanged. This enables modular system architectures with higher robustness and increased privacy. Furthermore, it brings higher scalability; in contrast, most previous approaches require the knowledge of the entire catalog, which is huge in practice. Our initial investigation shows that while this may come at a cost of performance (i.e., knowing the entire catalog could lead to optimal performance), a good trade-off between performance and scalabitity is feasible in practice (Appendix~\ref{sec:results-greedy}).} 

\myitem{Techno-economic feasibility.} \mjr{The existing approaches in} joint network and recommendations paradigm mainly assume that the content provider (CP) and the network operator (NO) are the same entity or closely collaborate and exchange information. The convergence between CPs and NOs is enabled due to the architectural developments of MEC and RAN Sharing~\cite{liang2015wireless}, while CPs increasingly deploy their own infrastructure to bring content closer to the user, e.g., Netflix OpenConnect, Google Global Cache, or bring their equipment inside the network of NOs~\cite{akamai-mobile-optim}. However, this collaboration requires some investment in infrastructure and technology, and changes in business strategy. In this context, our modular/black-box approach could lower the barrier for techno-ecomonic feasibility of the joint network and recommendations paradigm: it does not rely on the exchange of private/sensitive information, and thus can (i) be applicable even when the network and the RS are not controlled by the same entity, and (ii) cope with potential tussles between CPs and NOs.

\subsection{Network-aware Recommendations with \ourAlgo}\label{sec:cabaret}

\blue{We now proceed to apply the proposed methodology in the real service of YouTube. To this end, we design a network-aware recommendation algorithm (\ourAlgo) that leverages information provided by the YouTube RS (\secref{sec:detailed-algo}), and discuss the related design implications (\secref{sec:tune-algo}).}

\textit{Remark:} We would like to stress that while we focus on the YouTube service, our approach and the \ourAlgo algorithm are generic and can be applicable to other video/radio services. \mjr{For example, the majority of popular content services provide public APIs, such as Vimeo~\cite{vimeo-api}, Twitch~\cite{twitch-api}, Dailymotion~\cite{dailymotion-api}, Spotify~\cite{spotify-api}, or in case APIs are not available, indirect methods (e.g., web-based parsing/crawling methods) exist for retrieving content recommendations, e.g., for Netflix~\cite{neflix-unofficial-search} or for Facebook through the \textit{Tracking Exposed} project~\cite{tracking-exposed}.}

\subsubsection{The \ourAlgo Algorithm}\label{sec:detailed-algo}

\myitem{\ourAlgo overview}. \blue{\ourAlgo receives information about content relations from the YouTube RS through its API. In particular, when a user watches a video $v$, \ourAlgo requests from the YouTube API a list of video IDs $\mathcal{L}$ related to $v$, i.e., the videos that YouTube would recommend to the user. Then it requests the related video IDs for every video in $\mathcal{L}$ and adds them in the end of $\mathcal{L}$, and so on, in a Breadth-First Search (BFS) manner. In the end of the process, the list $\mathcal{L}$ contains IDs of videos directly and indirectly related to $v$; of these videos, the top $N$ that are cached and/or highly related to $v$ are finally recommended to the user. 
An example is depicted in Fig.~\ref{fig:rec-algo}.}

\begin{figure}
\centering
\includegraphics[width=1\linewidth,height=0.6\linewidth]{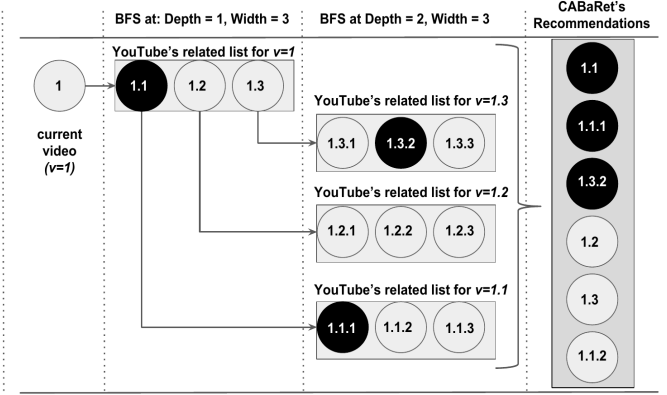}
\caption{CABaRet: example with $D_{BFS}=2$, $W_{BFS}=3$, $N=6$. Cached videos are denoted with black color.}
\label{fig:rec-algo}
\end{figure}

\myitem{Input.} The recommendation algorithm receives as input:
\begin{itemize}[leftmargin=*]
\item $v$: the video ID (or URL) which is currently watched
\item $N$: the number of videos to be recommended
\item $\mathcal{C}$: the list with the IDs of the cached videos
\item $D_{BFS}$: the depth to which the BFS proceeds
\item $W_{BFS}$: the number of related videos that are requested \textit{per content} from the YouTube API (i.e., the ``width'' of BFS)
\end{itemize}

\myitem{Output.} The recommendation algorithm returns as output:
\begin{itemize}[leftmargin=*]
\item $\mathcal{R}$: ordered list of $N$ video IDs to be recommended.
\end{itemize}
\myitem{Workflow.} \ourAlgo searches for videos related to video $v$ in a BFS manner as follows (\textit{line 1} in Algorithm~\ref{alg:recommendation}). Initially, it requests the $W_{BFS}$ videos related to $v$, and adds them to a list $\mathcal{L}$ in the order they are returned from the YouTube API. For each video in $\mathcal{L}$, it further requests $W_{BFS}$ related videos, as shown in Fig.~\ref{fig:rec-algo}, and adds them in the end of $\mathcal{L}$. It proceeds similarly for the newly added videos, until the depth $D_{BFS}$ is reached; e.g., if $D_{BFS}=2$, then $\mathcal{L}$ contains $W_{BFS}$ video IDs related to $v$, and $W_{BFS}\cdot W_{BFS}$ video IDs related to the related videos~of~$v$.

Then, \ourAlgo searches for video IDs in $\mathcal{L}$ that are also included in the list of cached videos $\mathcal{C}$ 
and adds them to the list of video IDs to be recommended $\mathcal{R}$, until all IDs in $\mathcal{L}$ are explored or the list $\mathcal{R}$ contains $N$ video IDs, whichever comes first (\textit{lines 4--9}).
If after this step, $\mathcal{R}$ contains less than $N$ video IDs, $N-|\mathcal{R}|$ video IDs from the head of the list $\mathcal{L}$ are added to $\mathcal{R}$; these IDs correspond to the top $N-|\mathcal{R}|$ non-cached videos that are directly related to video $v$ (\textit{lines 10--15}). \mjr{\textit{Remark:} the operations in lines 4--9 and 10--15 could be merged in an implementation to slightly reduce complexity.}


\myitem{Extension: different costs per video.} \mjr{In more generic setups, each video $i$ may have a different delivery cost $c_{i}$. \ourAlgo can be easily modified for this case, by selecting the $N$ video with the lowest costs $c_{i}$ (and prioritizing videos found earlier in the BFS among those with equal costs).}

\begin{algorithm}
\begin{small}
\begin{algorithmic}[1]
\caption{\\~\textit{CABaRet}: Cache-Aware \& BFS-related  Recommendations}\label{alg:recommendation}
\Statex {$Input: v, N, \mathcal{C}, D_{BFS}, W_{BFS}$} 
\State $\mathcal{L}\gets BFS(v,D_{BFS},W_{BFS})$ \Comment ordered set of video IDs
\State $\mathcal{R}\gets \emptyset$ \Comment ordered set of video IDs to be recommended
\State $i\gets 1$
\For {$c\in \mathcal{L}$}
	\If {$i\leq N$ and $c\in \mathcal{C}$}
		\State $\mathcal{R}.{append}(c)$
		\State $i\gets i+1$
	\EndIf
\EndFor
\For {$c\in \mathcal{L}\setminus \mathcal{R}$}
	\If {$i\leq N$}
    	\State $\mathcal{R}.{append}(c)$
		\State $i\gets i+1$
    \EndIf
\EndFor

\State $return~~\mathcal{R}$
\end{algorithmic}
\end{small}
\end{algorithm}

\subsubsection{Implications and Design Choices}\label{sec:tune-algo}

\myitem{High-quality recommendations.} Using the baseline RS (here, the YouTube recommendations) ensures strong relations between videos that are directly related to $v$ (i.e., BFS at depth 1). Moreover, typically the baseline RS provides only a subset of the relevant recommendations to the user; for example, while the YouTube RS finds hundreds of videos highly related to $v$, only a few of them (e.g,. 5 or 20, depending on the end device) are finally communicated to the user~\cite{covington2016deep}. The rationale behind our methodology and \ourAlgo is to explore the related videos that are not communicated to the user. To this end, based on the fact that related videos are similar and have high probability of sharing recommendations (i.e., if video $a$ is related to $b$, and $b$ to $c$, then it is probable that $c$ relates to $a$)~\cite{survey-collaborative-filtering,amazon-recommendations}, \ourAlgo tries to infer these latent video relations through BFS. Hence, videos found by BFS in depths $>1$ are also (indirectly) related to $v$ and probably good recommendations as well.

To further support the above claim, we collect and analyze datasets of related YouTube videos. Specifically, we consider the set of most popular  videos, denoted as $\mathcal{P}$, in a region, and for each $v\in\mathcal{P}$ we perform BFS by requesting the list of related videos (similarly to \textit{line 1} in \ourAlgo). We use as parameters $W_{BFS}=\{10,20,50\}$ and $D_{BFS}=2$, i.e., considering the directly related videos (depth 1) and indirectly related videos with depth 2. We denote as $\mathcal{R}_{1}(v)$ and $\mathcal{R}_{2}(v)$ the set of videos found at the first and second depth of the BFS, respectively. We calculate the fraction of the videos in $\mathcal{R}_{1}(v)$ that are also contained in $\mathcal{R}_{2}(v)$, i.e., $I(v) = \frac{| \mathcal{R}_{1}(v)  \cap \mathcal{R}_{2}(v)|}{|\mathcal{R}_{1}(v)|}$. High values of $I(v)$ indicate a strong similarity of the initial content $v$ with the set of indirectly related contents at depth 2.

Table~\ref{table:I-v} shows the median values of $I(v)$, over the $|\mathcal{P}|=50$ most popular contents in the region of Greece (GR), for different BFS widths. As it can be seen, $I(v)$ is very high for most of the initial videos $v$. For larger values of $W_{BFS}$, $I(v)$ increases, and when we fully exploit the YouTube API capability, i.e., for $W_{BFS}$=50, which is the maximum number of related videos returned by the YouTube API, the median value of $I(v)$ becomes larger than $0.9$. Finally, we measured the $I(v)$ in other regions as well, and observed that even in large (size/population) regions, the $I(v)$ values remain high, e.g., in the United States (US) region, $I(v)$=0.8 for $W_{BFS}$=50.
\begin{table}[h]
\centering
\caption{$I(v)$ vs. $W_{BFS}$ for the region of GR
.}
\label{table:I-v}
\begin{tabular}{|cccc|
}
\hline
{$W_{BFS}:$}	&{10}&{20}&{50}
\\
{$I(v):$}	&0.70&0.85&0.92
\\
\hline
\end{tabular}
\vspace{-\baselineskip}
\end{table}




\myitem{Tuning \ourAlgo.} 
\mjr{Typically, users prefer videos in the top of the recommendation list, hence, \ourAlgo puts in the top of the list $\mathcal{R}$ the cached videos found in the BFS\footnote{\mjr{Nevertheless, if for a service the patterns of users preferences is different (e.g., preference to the bottom of recommendation list), \ourAlgo could be tuned accordingly.}}.} 

Moreover, the parameters $D_{BFS}, W_{BFS}$ can be tuned to achieve a desired performance, e.g., in terms of probability of recommending a cached or highly related video.
For large $D_{BFS}$, the similarity between $v$ and the videos \textit{at the end of the list} $\mathcal{L}$ is expected to weaken, while for small $D_{BFS}$ the list $\mathcal{L}$ is shorter and it is less probable that a cached content is contained in it. Hence, the parameter $D_{BFS}$ can be used to achieve a trade-off between quality of recommendations (small $D_{BFS}$) and probability of recommending a cached video (large $D_{BFS}$). The number of related videos requested per content $W_{BFS}$, can be interpreted similarly to $D_{BFS}$. A small $W_{BFS}$ leads to considering only top recommendations per video, while a large $W_{BFS}$ leads to a larger~list~$\mathcal{L}$. \blue{For the size of the list $\mathcal{L}$ it holds that 
\[\textstyle|\mathcal{L}|\leq \sum_{n=1}^{D_{BFS}}(W_{BFS})^{n}\]
where the equality holds when all videos found by the BFS are unique.}

\textit{Remark:} YouTube imposes quotas on the API requests per application per day, which prevents API users from setting the parameters $W_{BFS}$ and $D_{BFS}$ to arbitrarily large values. However, even with small number of API requests (for related contents), the exploration returns a large number of \textit{unique} videos. Figure~\ref{fig:retVSreq} shows how many relations are requested for parameters $W_{BFS}\in \{1,...,50\}$ and $D_{BFS}=2$ (x-axis), versus the size of the returned list $|\mathcal{L}|$ (y-axis). Two settings are considered, where the BFS starts from a top trending YouTube video from the YouTube ``front page'' or from a video searched through the ``search bar'' (see details in~\secref{sec:measurements}). 
\mjr{In both cases, and as already suggested by the results of Table~\ref{table:I-v}, the BFS discovers several duplicates. On the one hand, this indicates a high-quality of recommendations. On the other hand, the number of unique video IDs in the list $\mathcal{L}$, increases linearly or almost linearly with the number of explored relations, thus indicating that the BFS achieves an efficient exploration (large $|\mathcal{L}|$).}

\begin{figure}
\centering
    \subfigure[``Front Page'' video demand]{\includegraphics[width=0.49\columnwidth]{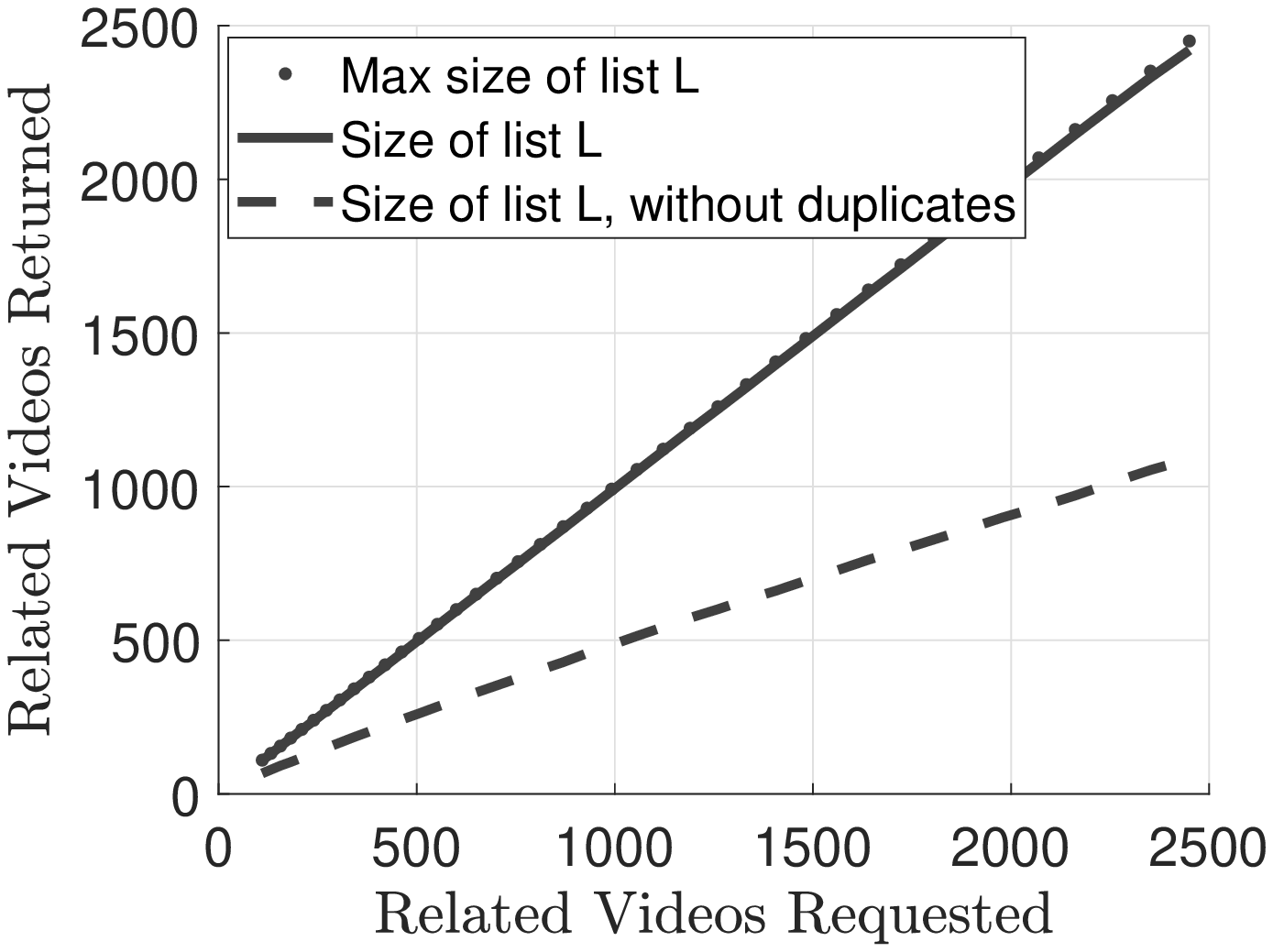}}
\subfigure[``Search Bar'' video demand]{\includegraphics[width=0.49\columnwidth]{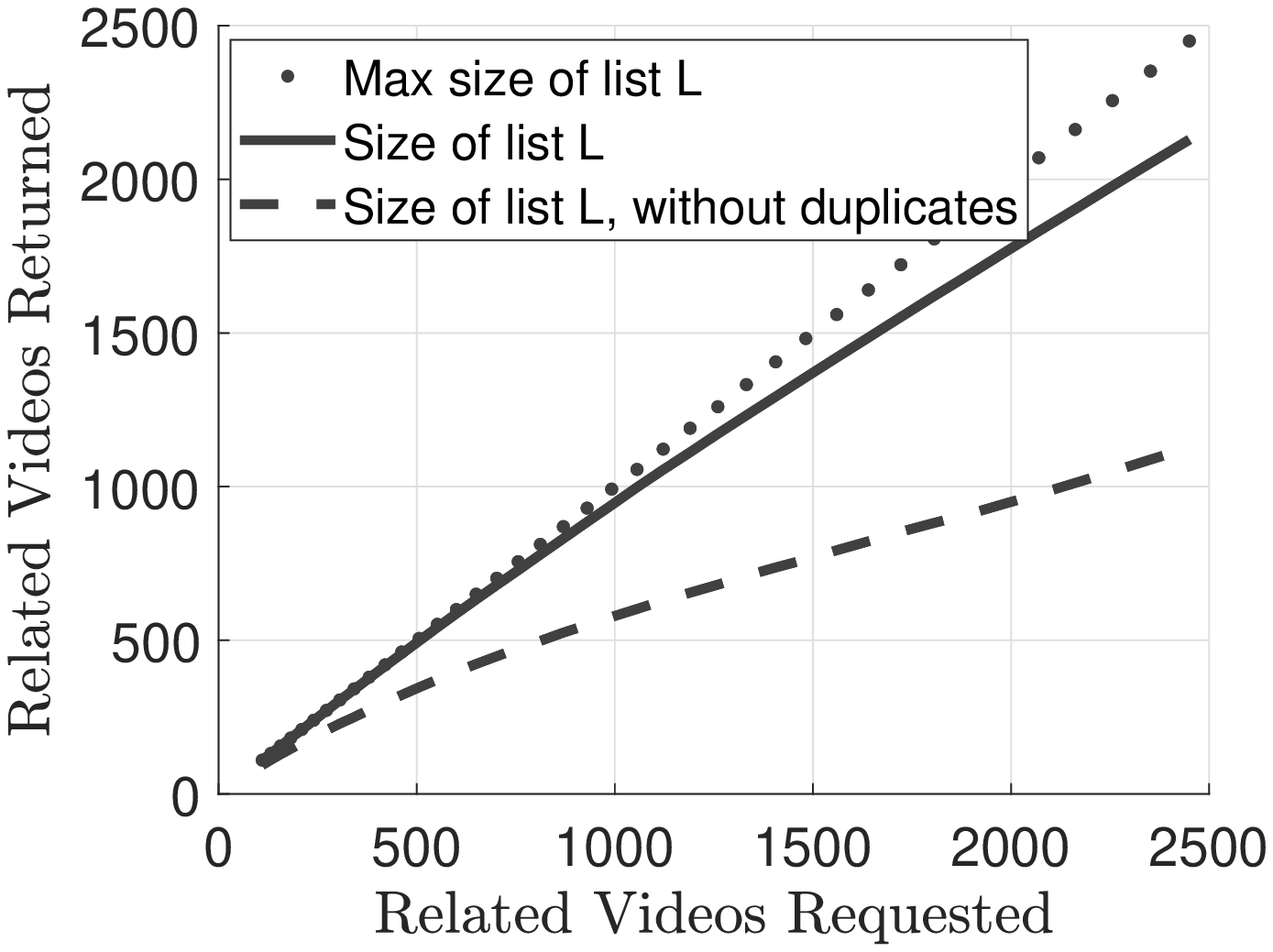}}
\caption{Num. of requested related videos VS Num. of returned related videos from the API.
}
\label{fig:retVSreq}
\end{figure}

\mjr{Finally, $N$ can be selected to fit different device or application settings (e.g., browser or mobile app), but also affects the performance; e.g., if $N$ is small, the user has a few options, which can further promote cached videos but decrease the quality of recommendations.}

In practice, \ourAlgo can be fine-tuned through experimentation with real users, e.g., A/B testing iterations, which is a common approach for tuning recommendation systems~\cite{covington2016deep}.

\myitem{\mjr{Performance modeling.}} \mjr{The performance of \ourAlgo can be measured by the fraction of requests made for the cached videos (i.e., in our setup, the \textit{cache hit ratio} or \textit{CHR}). This depends on how many cached videos are recommended to the user and the probability a user to select one of them; in practice, these quantities are intertwined and depend on complex user demand patterns. In the following sections, we conduct extensive measurements and experiments to quantify the achieved performance in the YouTube service and under realistic user demand patterns. However, here, we also provide an analytical model to predict the achieved CHR, which can be applied to any service and can be useful to obtain initial performance estimations (e.g., before proceeding to measurements for a more detailed evaluation).}

\mjr{Let us denote the probability that a user selects a video at the $i^{th}$ position ($i=1,...,N$) of the recommendation list as $p_{i}$. Also let $M$ be the number of cached contents that is found in the BFS, i.e.,  $M=|\mathcal{L}\cap\mathcal{C}|$, and thus the number of cached videos in the recommendation list of \ourAlgo is $\max\{M,N\}$. Then, the fraction of requests for cached videos (CHR) will be $\sum_{i=1}^{\max\{M,N\}}p_{i}$, and taking the expectation over $M$, gives:
\[\textstyle CHR=\sum_{m=1}^{|\mathcal{L}|}\left(\sum_{i=1}^{\max\{m,N\}} p_{i}\right) \cdot P\{M=m\}\]
where $M$ follows a Binomial distribution with $|\mathcal{L}|$ number of trials and success probability $q_{C}$, where $q_{C}$ is the probability that a recommendation is for a cached content\footnote{\mjr{Typically, (i) the most popular contents are cached and (ii) recommendations have bias towards popular contents (popularity bias~\cite{steck2011item,abdollahpouri2017controlling,vall2019order}), which leads to high $q_{C}$, and thus high CHR.}}.}

\mjr{For the most common case of $p_{i}\geq p_{j}$ for $i<j$ (i.e.,  users preference is higher for top recommendations), the inner sum in the above expression is a concave function of $m$. Thus,
Jensen's inequality allows us to upper bound the CHR by a simpler expression involving only the 
%
mean number of cached contents found by the BFS $\bar{M}$:~\footnote{\mjr{If user preferences are for recommendations at the end of the list ($p_{i}\leq p_{j}$ for $i<j$), Jensen's inequality gives $CHR\geq \sum_{i=1}^{\max\{\lfloor \bar{M} \rfloor,N\}} p_{i}$}}
$CHR\leq \sum_{i=1}^{\max\{\lceil \bar{M} \rceil,N\}} p_{i}
$. This bound is tight for our actual measurement results in Section~\ref{sec:measurements}.
}

\mjr{Despite the assumptions made in the model (e.g., independence between $p_{i}$ and the set of recommended videos), it can be generalizable to any service (e.g., to short-video services having considerably different user demand patterns than YouTube~\cite{li2020leveraging,zhang2019challenges_short_video}) given general user demand statistics, i.e., $q_{C}$ and $p_{i}$.
}

\myitem{Caching optimization under \ourAlgo.} \blue{\ourAlgo receives as input a list of cached videos $\mathcal{C}$ (or, more general, videos that can be delivered by the network in high quality) and returns cache-aware recommendations to increase the caching efficiency. Depending on the considered scenario, it may be possible to control the list $\mathcal{C}$ as well. Carefully selecting the contents in the list $\mathcal{C}$ can lead to further increase of the caching efficiency~\cite{chatzieleftheriou2019jointly,costantini2019approximation}. Under \ourAlgo recommendations it is possible to design the caching policy as well, so that it further increases the cache hit ratio as we showed in our preliminary work~\cite{kastanakis-cabaret-mecomm-2018}. While a detailed investigation is out of the scope of this paper, we provide in Appendix~\ref{appendix:joint-cache-rec} a formulation of the optimization problem, an approximation algorithm, as well as evaluation results for the extra increase that can be achieved by jointly selecting the caching policy under \ourAlgo.}

%% file: Measurements.tex
Using the proposed methodology and the \ourAlgo algorithm, we conduct \blue{extensive} measurements and experiments over the YouTube service\footnote{Our experiments and use of the YouTube API conform to the YouTube terms of service \url{https://www.youtube.com/static?template=terms}.}, to investigate the performance (in terms of cache hit ratios) of network-aware recommendations in MEC scenarios. The setup of the scenarios is presented in \secref{sec:experiments-setup}, and the results in~\secref{sec:results} \blue{and~\secref{sec:demand-google} for two video demand types
.}

\subsection{Setup}
\label{sec:experiments-setup}

\myitem{The YouTube API} provides a number of functions to retrieve information about videos, channels, user ratings, etc. In our measurements, we request 
the following information:
\begin{itemize}
\item the most popular videos in a region (max. 50)
\item the list of related videos (max. 50) for a given video
\end{itemize}
\textit{Remark}: In the remainder, we present results 
for the region of Greece (GR). Nevertheless, our insights hold also in the other regions we tested, and, indicatively, we briefly state results for the region of United States (US).



\myitem{Caching.} We assume a MEC cache storing the most popular videos in a region. Unless otherwise stated, we populate the list of cached contents with the top $C$ 
video IDs returned from the YouTube API.

\myitem{Recommendations.} We consider two classes of scenarios with (i) YouTube and (ii) \ourAlgo recommendations. In both cases, when a user enters the UI, the $50$ most popular videos in her region are recommended to her (as in YouTube's front page). Upon watching a video $v$, a list of $N=20$ videos is recommended to the user; the list is (i) composed of the top $N$ directly related videos returned from the YouTube API (YouTube scenarios), or (ii) generated by \ourAlgo with parameters $N$, $W_{BFS}$ and $D_{BFS}$ (\ourAlgo scenarios).

\myitem{Video Demand.} In each experiment, we assume a user that enters the UI and \blue{selects an initial video to watch in one of the following ways: (a) ``front-page recommendations'': the user selects to watch one of the initially recommended (i.e., $50$ most popular) videos recommended in the front page; or (b) ``search bar'': the user types in the search bar a keyword of her interest, and selects one of the returned video recommendations. These two types of initial requests represent the two most common ways of user behavior (note that the former captures also trending videos selections)~\cite{RecImpact-IMC10}. We present the results for each of the aforementioned initial video demand types separately, in~\secref{sec:results} and~\secref{sec:demand-google}, respectively; the former is expected to have a more concentrated demand among the most popular videos (and thus, higher CHR, since those are assumed to be cached), while the latter a more varying demand that stresses the caching system.} 

\blue{After the initial video}
, the system recommends a list of $N$ videos ($r_{1},r_{2},...,r_{N}$), and the user selects with probability $p_{i}$ to watch $r_{i}$ next. We set the probabilities $p_{i}$ to depend on the order of appearance --and not the content--
and consider \textit{uniform} ($p_{i}=\frac{1}{N}$) and \textit{Zipf} ($p_{i}\sim\frac{1}{i^{\alpha}}$) scenarios; the higher the exponent $\alpha$ of the Zipf distribution, the more preference is given by the user to the top recommendations (user preference to top recommendations has been observed in YouTube traffic~\cite{cache-centric-video-recommendation}).


\subsection{Results: ``Front-Page'' Video Demand}\label{sec:results}
\subsubsection{Single Requests}\label{sec:results-single}

We first consider scenarios of single requests (similarly to~\cite{sermpezis-sch-globecom,chatzieleftheriou2017caching}). In each experiment $i$ ($i=1,...,M$) a user watches one of the top popular videos, let $v_{1}(i)$, and then follows a recommendation and watches a video $v_{2}(i)$. We measure the Cache Hit Ratio (CHR), which we define as the fraction of the \textit{second requests} of a user that are for a cached video (since the first request is always for a cached --top popular-- video):
\begin{equation}\label{eq:chr-single-request}
\textstyle CHR = \frac{1}{M}\cdot \sum_{i=1}^{M} \mathbb{I}_{v_{2}(i)\in \mathcal{C}}
\end{equation}
where $\mathbb{I}_{v_{2}(i)\in \mathcal{C}} =1$ if $v_{2}(i)\in \mathcal{C}$ and $0$ otherwise, and $M$ the number of experiments\footnote{We considered all possible experiments on the collected dataset.}.

\myitem{CHR vs. BFS parameters.} Fig.~\ref{fig:chr-vs-bfs-param} shows the CHR achieved by \ourAlgo under various parameters, along with the CHR under regular YouTube recommendations, when caching all the most popular videos ($|\mathcal{C}|$=50). The efficiency of caching significantly increases with \ourAlgo, even when only directly related contents are recommended ($D_{BFS}$=1), i.e., \textit{without loss in recommendation quality}. Just reordering the list of YouTube recommendations (as suggested in~\cite{cache-centric-video-recommendation}), brings gains when $p_{i}$ is not uniformly distributed. However, the added gains by our approach are significantly higher. As expected, the CHR increases for larger $W_{BFS}$ and/or $D_{BFS}$; e.g., \ourAlgo for $W_{BFS}$=50 and $D_{BFS}$=2, achieves 8 to 10 times 
higher CHR than regular YouTube recommendations. Also, the CHR increases for more skewed $p_{i}$ distributions, since top recommendations are preferred and \ourAlgo places cached contents at the top of the recommendation list.

\begin{figure}\centering
\begin{minipage}[t]{0.46\linewidth}
\centering
\includegraphics[width=1\columnwidth]{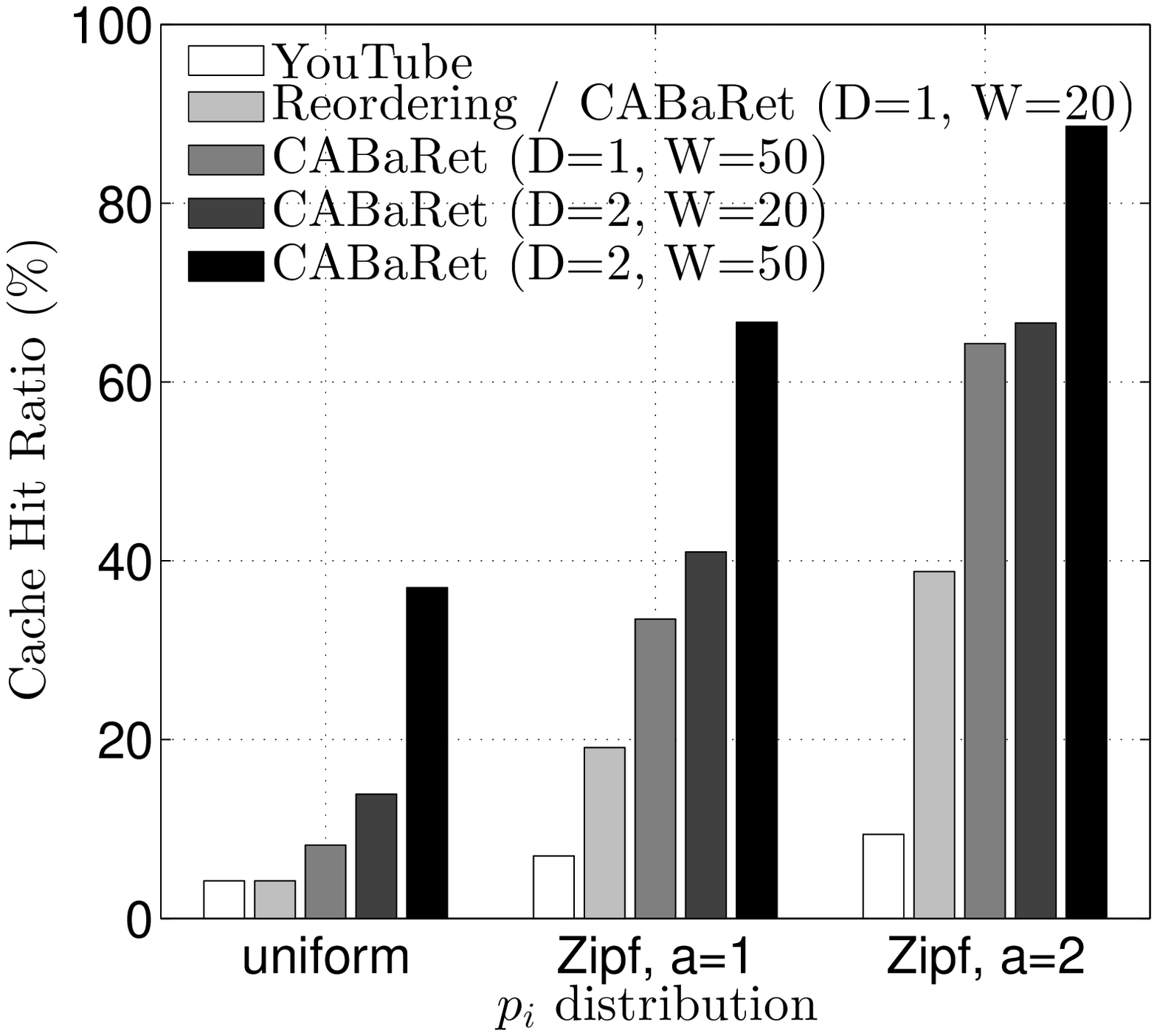}
\caption{CHR under different BFS parameters.
}
\label{fig:chr-vs-bfs-param}
\end{minipage}
\hspace{0.03\linewidth}
\begin{minipage}[t]{0.49\linewidth}
\centering%
\mjrout{%
\includegraphics[width=1\columnwidth]{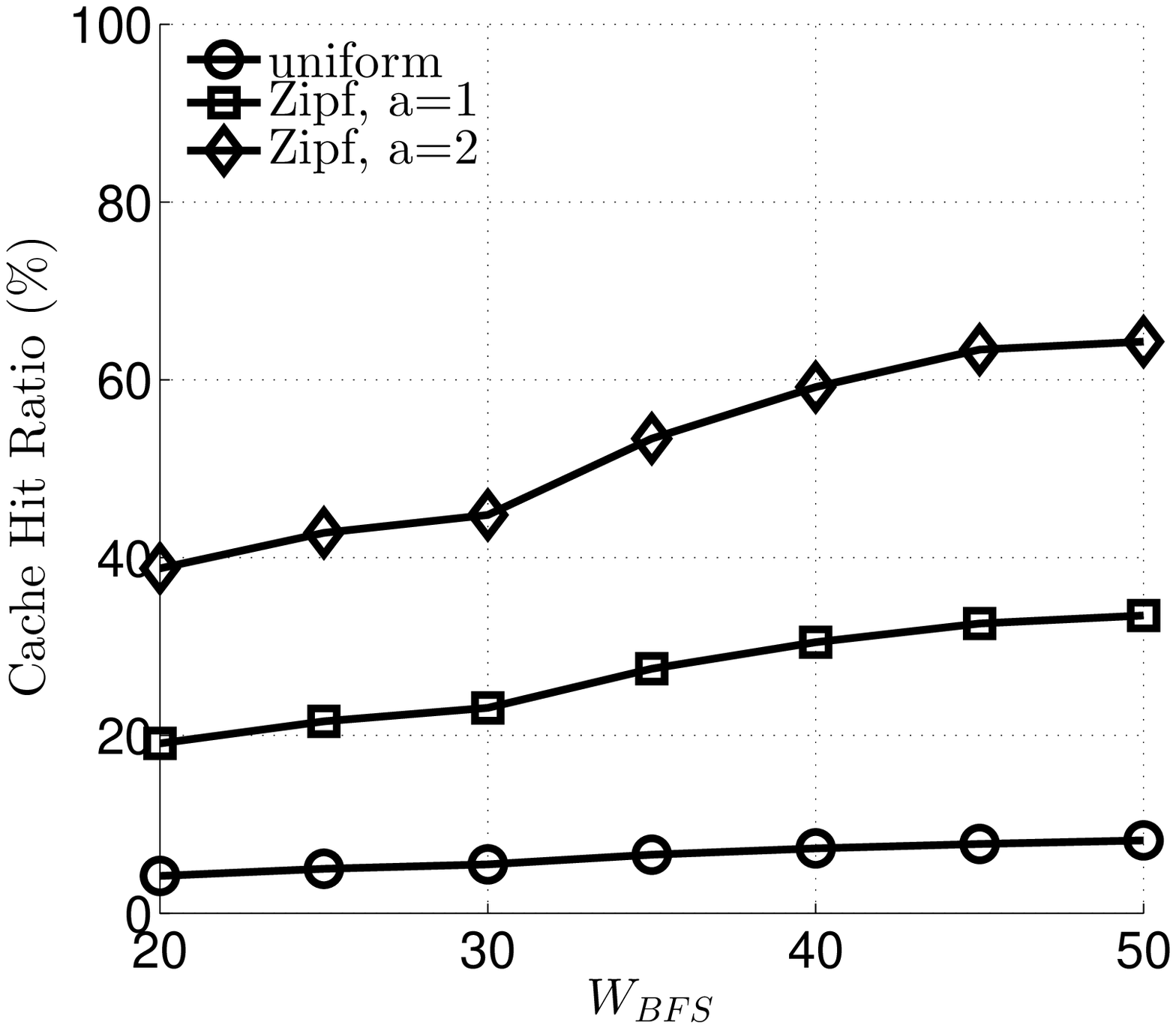}
\caption{\blue{CHR vs. $W_{BFS}$ ($D_{BFS}$=1).}}
}%
\includegraphics[width=1\columnwidth]{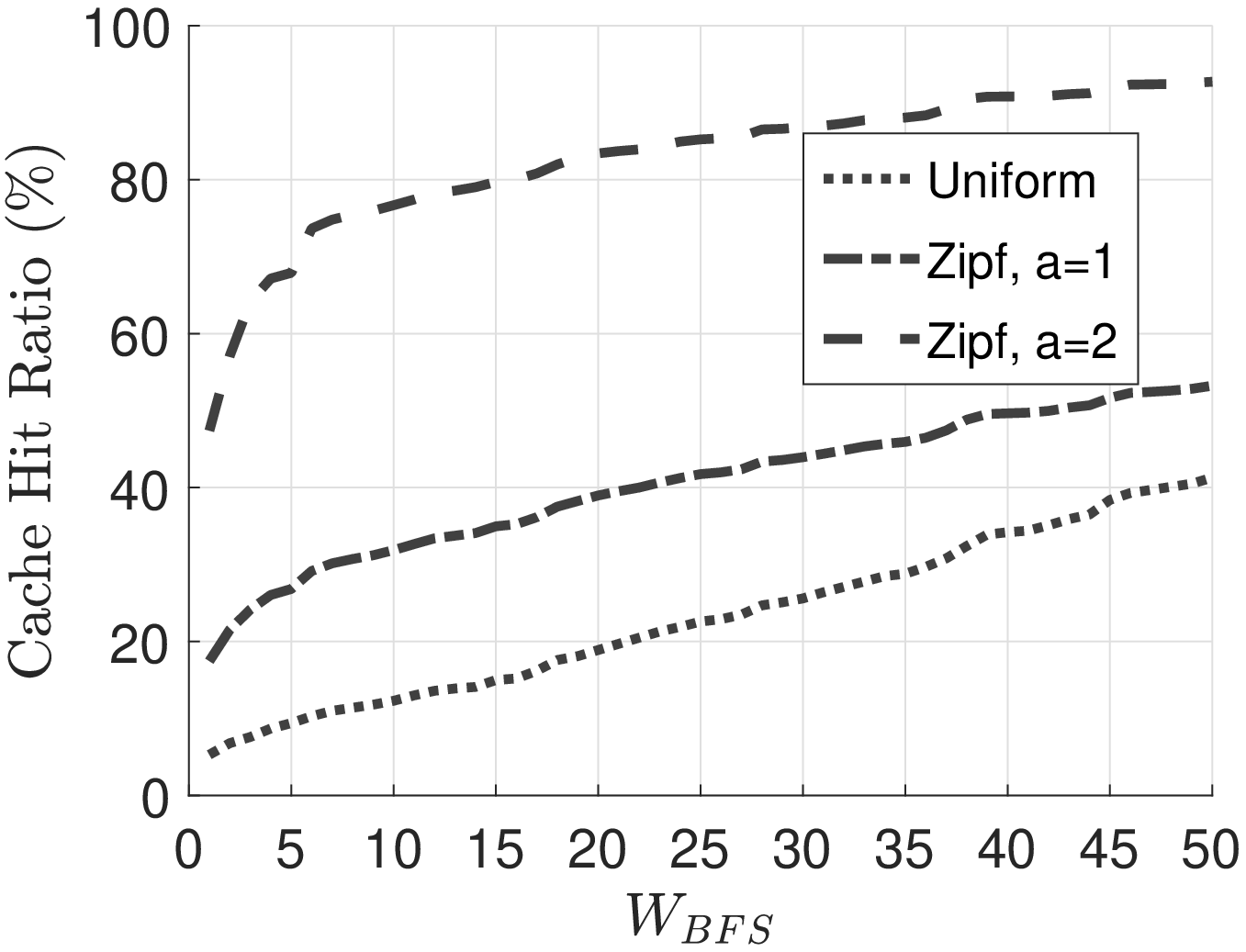}
\caption{\mjr{CHR vs. $W_{BFS}$ ($D_{BFS}$=2).}}
\label{fig:chr-vs-W}
\end{minipage}
\end{figure}

\mjrout{In Fig.~\ref{fig:chr-vs-W} we observe that, when higher preference is given to top recommendations, increasing the $W_{BFS}$ can increase CHR even with $D_{BFS}$=1. However, for uniform selection, significant gains cannot be brought only by increasing $W_{BFS}$. This indicates that, in such cases, considering also indirectly related contents (i.e., $D_{BFS}>1$, as \ourAlgo does), is required to yield a non-negligible increase in CHR (cf. Fig.~\ref{fig:chr-vs-bfs-param}).} 

In experiments concerning the --larger-- US region, the CHR values are lower for both regular YouTube ($<0.5\%$) and \ourAlgo ($1\%-43\%$) recommendations, due to the fact that the top popular videos appear with lower frequency in the related lists. However, the \textit{relative gains} from \ourAlgo are consistent with (or even higher than) the presented results.

\myitem{\mjr{CHR vs. knowledge of content relationships.}} \mjr{\ourAlgo uses only partial knowledge (i.e., black-box) of content relationships. This could bring some reduction in the maximum gains that can be achieved by a network-aware RS (knowledge vs. performance trade-off). For example, previously proposed algorithms that assume knowledge of the entire content relationships graph (which is equivalent to \ourAlgo with a large enough parameter $W_{BFS}$ to explore the entire catalog) could achieve higher gains. To quantify this trade-off, we present in Fig.\ref{fig:chr-vs-W} the CHR of \ourAlgo (performance) vs. the $W_{BFS}$ parameter (knowledge of content relationships). As expected the CHR increases when more information about the content relationships is available (i.e., larger $W_{BFS}$). However, when $p_{i}$ follows a Zipf distribution, which is more common in practice, the effect of $W_{BFS}$ is less intense. This indicates that the benefits of the proposed black-box approach (see \secref{sec:methodology-overview}) can be combined with a performance that is comparable to approaches requiring more information about the content relationships.}

\begin{figure}
\begin{minipage}[t]{0.47\linewidth}
\centering
\includegraphics[width=1\columnwidth]{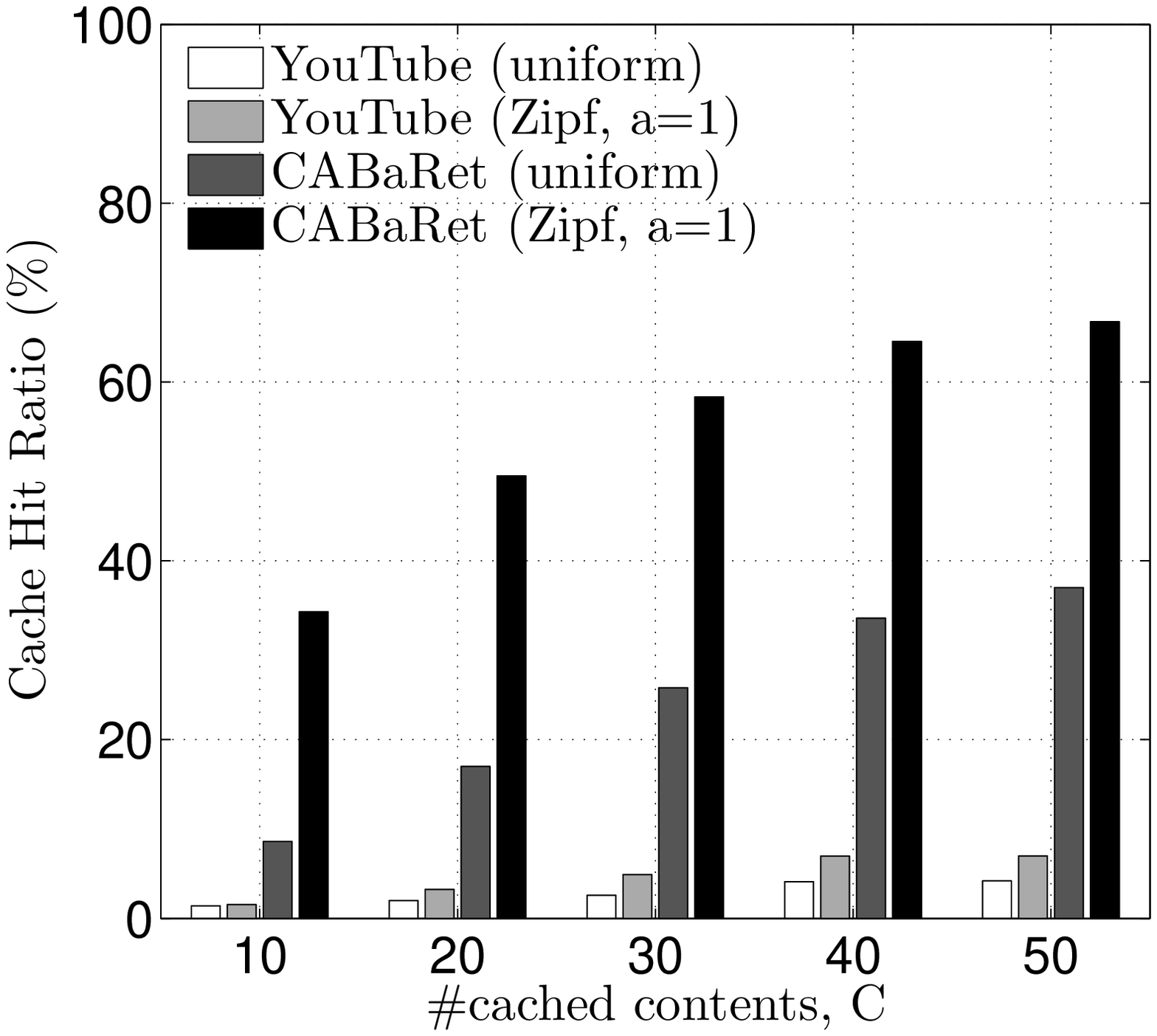}
\caption{CHR vs. \#~cached contents $C$\\ ($W_{BFS}$=50, $D_{BFS}$=2).}
\label{fig:chr-vs-C-d2}
\end{minipage}
\hspace{0.03\linewidth}
\begin{minipage}[t]{0.47\linewidth}
\centering
\includegraphics[width=1\columnwidth]{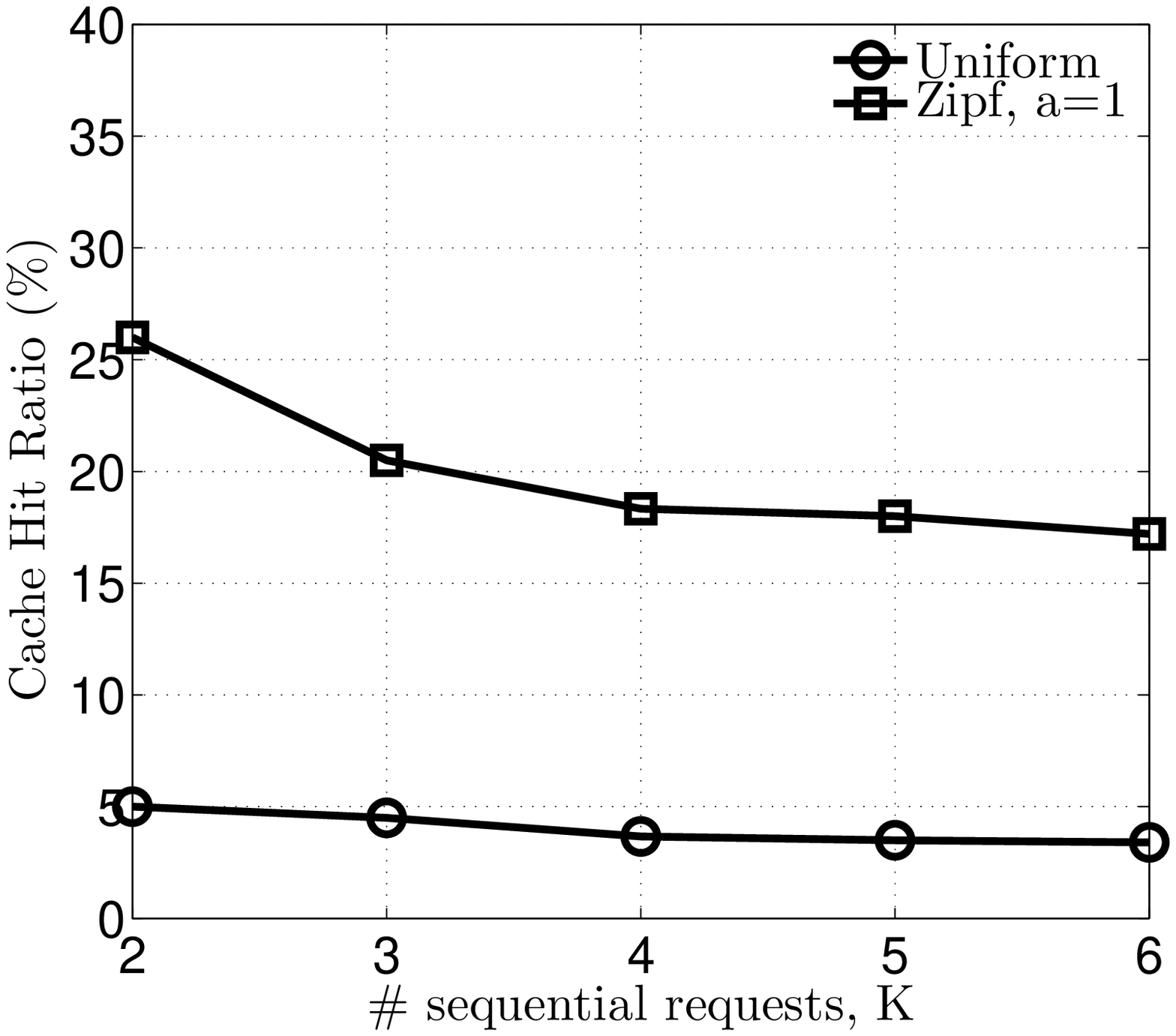}
\caption{CHR vs. \#~requests in sequence $K$ ($C$=20, $W_{BFS}$=20, $D_{BFS}$=2). \textit{Note}: y-axis up to 40\%.}
\label{fig:chr-vs-K-sequential}
\end{minipage}
\end{figure}

\myitem{CHR vs. number of cached videos.} We further consider scenarios with varying number of cached contents $C = |\mathcal{C}|$. In each scenario, we assume that the $C$ most popular contents are cached. Fig.~\ref{fig:chr-vs-C-d2} 
shows the CHR achieved by \ourAlgo, in comparison to scenarios under regular YouTube recommendations. The results are consistent for all considered values of $C$; the CHR under \ourAlgo is significantly higher than in the YouTube case. Moreover, even when caching a small subset of the most popular videos, \ourAlgo brings significant gains. E.g., by caching $C=10$ out of the $50$ top related contents \ourAlgo increases the CHR from $2\%$ and $3.2\%$ to $17\%$ and $50\%$, for the uniform and Zipf($\alpha$=1) scenarios, respectively.



\subsubsection{Sequential Requests}
We now test the performance of our approach in scenarios where users enter the system and watch a sequence of $K$, $K>2$, videos (similarly to~\cite{giannakas-wowmom-2018}, and in contrast to the previous case, where they watch only two videos, i.e., $K=2$). At each step, the system recommends a list of videos to the user by applying \ourAlgo on the currently watched video. We denote as $v_{k}(i)$ the $k^{th}$ video requested/watched by a user in experiment $i$. We measure the CHR, which is now defined as
\begin{equation}\label{eq:chr-sequential-request}
\textstyle CHR = \frac{1}{M}\cdot \sum_{i=1}^{M} \sum_{k=2}^{K} \mathbb{I}_{v_{k}(i)\in \mathcal{C}}
\end{equation}
where $\mathbb{I}_{v_{k}(i)\in \mathcal{C}} =1$ if $v_{k}(i)\in \mathcal{C}$ and $0$ otherwise, over $M=100$ experiments per scenario.

Moving ``farther'' from the initially requested video (which belongs to the list of most popular and cached videos) through a sequence of requests, we expect the CHR to decrease, due to lower similarity of the requested and cached videos. However, as Fig.~\ref{fig:chr-vs-K-sequential} shows, the decrease in the CHR (under \ourAlgo recommendations) is not large. The CHR remains close to the case of single requests (i.e., for $K$=2 in the x-axis), indicating that our approach performs well even when we are several steps far from the cached videos. In fact, caching more than the top most popular videos appearing on the front page, would further reduce the CHR decrease.


%% file: Demand_google_trends.tex
Up to now, we have considered a user that starts his/her viewing session by selecting one of the trending videos recommended in the YouTube homepage. While this is a common behavior (in YouTube and similar services), we now consider the other popular option for a user, which is to enter the YouTube webpage/app and select a desired video (e.g., through the search bar or directly typing the video url) irrespectively of the current trends. In the following, we describe our measurements and experiments for realistic scenarios that capture this second class of user behavior. Since considering users to select arbitrary initial videos, dramatically increases the set of initial videos (i.e., from $50$ top trending contents in~\secref{sec:results} to -theoretically- the entire YouTube catalogue that counts more than 5 billion videos), the CHR achieved by \ourAlgo (and any algorithm) is expected to decrease. Our goal here is to quantify the CHR gains when users start their session by searching a video through the search bar, and test whether the proposed approach can still provide considerable benefits in this ``worst-case'' scenario.

\textit{Remark:} The cache stores the $C$ top most popular YouTube videos, as in~\secref{sec:experiments-setup}. Hence, our results are comparable to the results of~\secref{sec:results}, and demonstrate the performance of the same scheme for this second class of users (who have different initial video demand).

\myitem{Initial video demand through the ``search bar''.}
We assume a user that enters the UI and searches though the search bar for a video according to her preferences (i.e., she does not watch one of the recommended trending videos as in~\secref{sec:experiments-setup}). While for the initial demand we could select randomly a video from the entire YouTube catalogue, e.g., uniformly or with a probability proportional to the total number of views, this would not capture the user behavior observed in practice: not all contents are equally probable to be selected, total number of views is not necessarily proportional to current demand (e.g., recent videos attract more clicks than older videos), timely topics attract more attention, etc. Hence, to simulate realistic ``video searches'' we apply the following methodology. 

\begin{itemize}[leftmargin=*]
    \item The user types a keyword/phrase in the search bar. To obtain a dictionary of keywords that correspond to popular and recent interests, we use the  Google Trends API~\cite{google-trends-api}. For each region, we collect the top 10 keywords for seven consequent days within a week 
    (some examples of keywords from our dataset are ``NBA Top Plays'', ``Avengers Trailer'', ``Grammy Nominations'', ``How to boil an egg?'', etc.).
    
    \item To map keywords (Google Trends) to YouTube videos, we pass each keyword to the YouTube API, which returns a list of video IDs, i.e., the list that would be returned if a user entered this keyword in the YouTube search bar. We select the first video ID from the list of each keyword. In total, we collect $70$ video IDs, of which we use the first $50$ (for consistency with the top $50$ trending videos in Section~\ref{sec:experiments-setup}). We call the list of these $50$ video IDs, as ``top Google trends''.
    
    \item In each experiment, we assume a user that enters the UI, watches one of the $50$ ``top Google trends'' videos, and then select one of the $N$ recommended videos to watch next (as described in Section~\ref{sec:experiments-setup}.
\end{itemize}

\myitem{CHR vs. BFS parameters.} Figure~\ref{fig:searchListResults} shows the CHR (single requests - \eq{eq:chr-single-request}) achieved with different \ourAlgo parameters in various scenarios (x-axis). The absolute values of CHR are in all scenarios 20\%--65\% lower compared to those in Fig.~\ref{fig:chr-vs-bfs-param}, which corresponds to users with initial requests for the top popular YouTube videos. While this decrease is expected in these more challenging scenarios (since the top YouTube videos are cached, but users start their viewing session from arbitrary videos), \ourAlgo can still effectively exploit the caching vectors and achieve, e.g., a CHR up to 32\% and 57\% in the Zipf(a=1) and Zipf(a=2) scenarios, respectively, in which otherwise we would observe a percentage within 1\%--2.4\% 
of cache hits under the original YouTube recommendations.

Moreover, we observe that the relative difference in performance between the different recommendation schemes (e.g., YouTube vs. \ourAlgo) remains the same as in Fig.~\ref{fig:chr-vs-bfs-param}: \ourAlgo significantly increases the caching efficiency by 10 times (uniform) to more than 20 times (Zipf, a=2) even for the more diverse (and thus challenging for the caching system) ``Search Bar'' video demand patterns.

\myitem{CHR vs. set of cached contents.} Figure~\ref{fig:searchListResults4} shows the CHR in the same video demand scenarios presented in Fig.~\ref{fig:searchListResults}, but now under a different caching policy. We consider a cache that stores $50$ videos of the ``top Google trends''. This makes the caching policy more targeted to the ``Search Bar'' traffic demand, and thus we expected an improved performance. Fig.~\ref{fig:searchListResults4} verifies this intuition: in all scenarios the CHR of \ourAlgo is higher than in Fig.~\ref{fig:searchListResults}. \blue{These results suggest that changing also the caching policy to better match the recommendations, i.e., joint selection of caching and recommendations, can further improve performance; indeed in Appendix~\ref{appendix:joint-cache-rec} we show that when caching is optimized under \ourAlgo recommendations, an extra $\times$2 increase in the cache hit ratio can be achieved.  
}

Another interesting observation in the top Google trends caching scenarios is that even with $W_{BFS}=20$ we can achieve high performance (CHR comparable to $W_{BFS}=50$), which was not the case in the other scenarios we tested (e.g., Fig.~\ref{fig:chr-vs-bfs-param} or Fig.~\ref{fig:searchListResults}). This indicates that the main factor to improve performance is the depth of the BFS: for $D_{BFS}=2$ the CHR becomes significantly higher; due to more diversity in ``Search Bar'' video demand and top Google trends, we need to explore deeper in the video relationship lists. 

\begin{figure}
\centering
\subfigure[Caching: top-50 most popular YouTube]{
\includegraphics[width=0.46\linewidth]{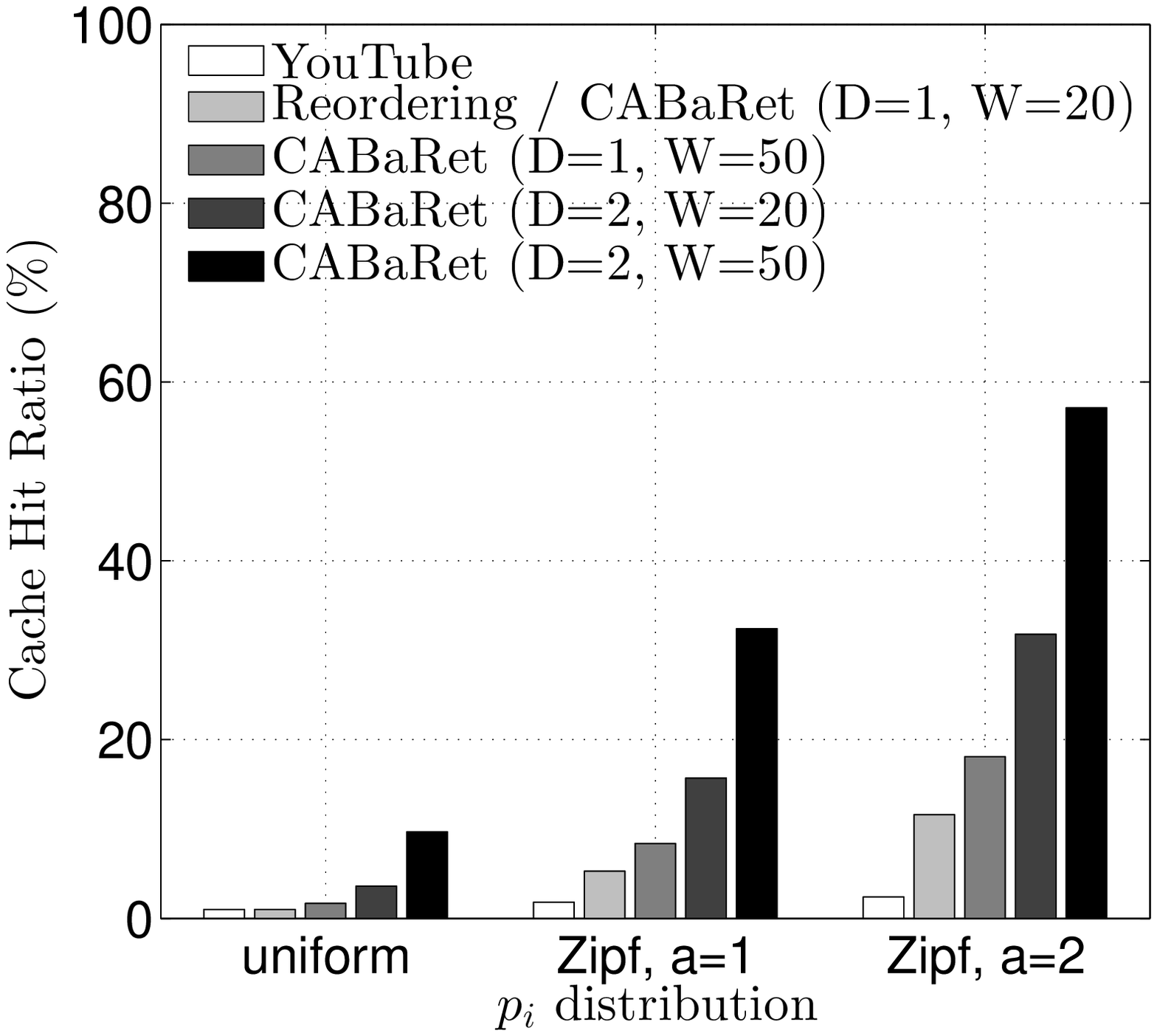}\label{fig:searchListResults}}
\subfigure[Caching: top-50 Google trends]{
\includegraphics[width=0.46\linewidth]{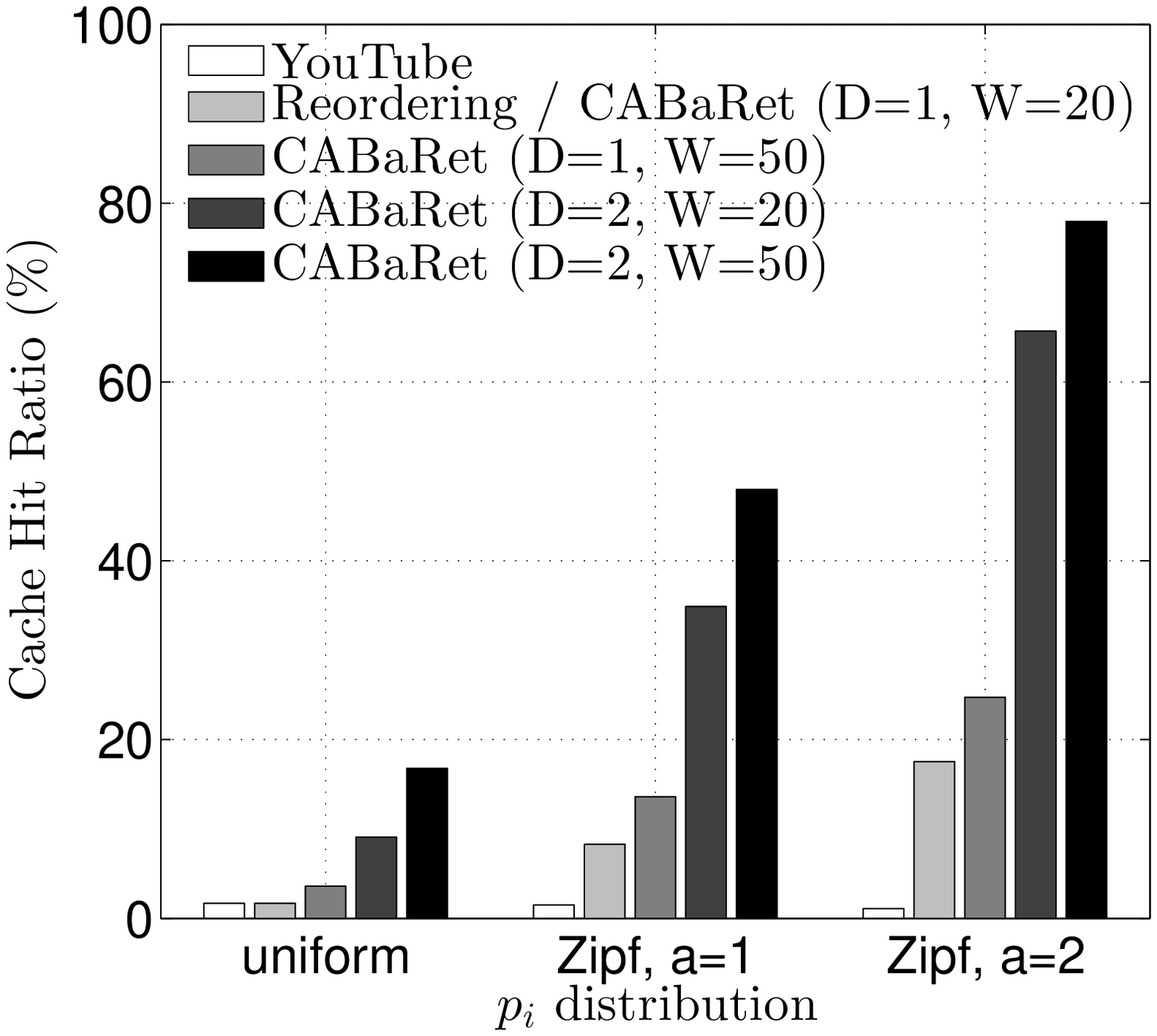}\label{fig:searchListResults4}}
\caption{CHR under different BFS for ``Search Bar'' video demand and cache populated with the (a) top-50 most popular videos, and (b) top-50 Google trends }
\end{figure}

%% file: ExperimentSetup.tex
We implemented an experimental platform with the architecture and main functionality of the framework presented in~\secref{sec:methodology}. Our goal is to conduct experiments with real users to (i) evaluate the performance that can be achieved in practice, and (ii) validate our assumptions, insights and measurement findings.

\myitem{Overview.} The UI is designed to accommodate our experiments (rather than resembling a real service or a prototype), and a screenshot is shown in Fig.~\ref{fig:exp-screenshot} (more details in~\secref{sec:experiment-session}). For the back-end, we assume that a list of cached video IDs is available at the time of the experiment (see~\secref{sec:experiments-setup}), and we use the YouTube API to embed a YouTube video player in our platform and serve video contents to the participants of the experiment. Finally, we generate recommendations using the \ourAlgo algorithm.


\myitem{Open-source code.} To facilitate future research on this topic, we open-source the code of the experimental testbed~\cite{cabaret-github}. Moreover, our implementation is modular and easily extensible. Thus, researchers and practitioners can use (as well as configure, parametrize, modify, or extend) our testbed to conduct their own experiments. More specifically: (i) the UI can be easily configured to present a desired number of recommendations $N$, include a search bar (e.g., to conduct experiments similar to~\secref{sec:demand-google}), add/remove rating questions, etc.; (ii) the list of cached video IDs in the back-end can be arbitrarily modified; (iii) the researcher can implement and use any other new algorithm (instead of \ourAlgo), by only modifying and calling a different method in the recommendation module.

\myitem{Collected dataset.} We conducted an experimental campaign recruiting participants through mailing lists and social media, and collected 742 samples from users in regions around the world. Adding to the open-source code, we also publish the dataset with the results of our experiments~\cite{cabaret-github}, which contains more information than those presented in this paper\footnote{We refer the interested reader to~\cite{sermpezis2019towards} for a more detailed analysis of the experimental results.}. We believe that this dataset can be of interest and facilitate researchers, since recruiting users and conducting experiments is an arduous task.

\subsection{Experiment Session}\label{sec:experiment-session}
We invited users to visit our platform and participate in our experiment. We first summarize here the steps of each experiment/session, and elaborate on some key steps subsequently. 
\begin{description}[font=\normalfont\itshape\space]
    \item[Action 1:] The user enters the platform and is requested to select from a list his/her preferred \textit{region}. 
    
    \item[Action 2:] After selecting a region, she is redirected to a page with instructions about the experiment. There, she is asked to start the viewing session by selecting a video from a list of $20$ trending (in the selected region) videos.
    
    \item[Action 3:] When selecting a video to watch, the user is redirected to a page as shown in Fig.~\ref{fig:exp-screenshot}, where: (a) The user watches the video (for as much time as she wants); (b) $5$ videos are recommended to the user to watch next; (c) the user is requested to provide some ratings about her viewing experience, including the relevance of recommendations (\textit{QoR}).
    
    \item[Action 4:] The user selects one of the $5$ recommended videos to watch next, and then step 3 is repeated. The maximum number of videos to watch is $5$. After the fifth video, the experiment session ends.
\end{description}

The information that is communicated to the users (when they enter the experimental platform) is that they are going to select, watch, and rate a series of five YouTube videos for the purposes of a research study. No further information is revealed to users about how we select the videos to recommend, to avoid biasing their selections and ratings. We also inform the users that no personal information is collected.

\begin{figure}
    \centering
    \includegraphics[width=1\linewidth]{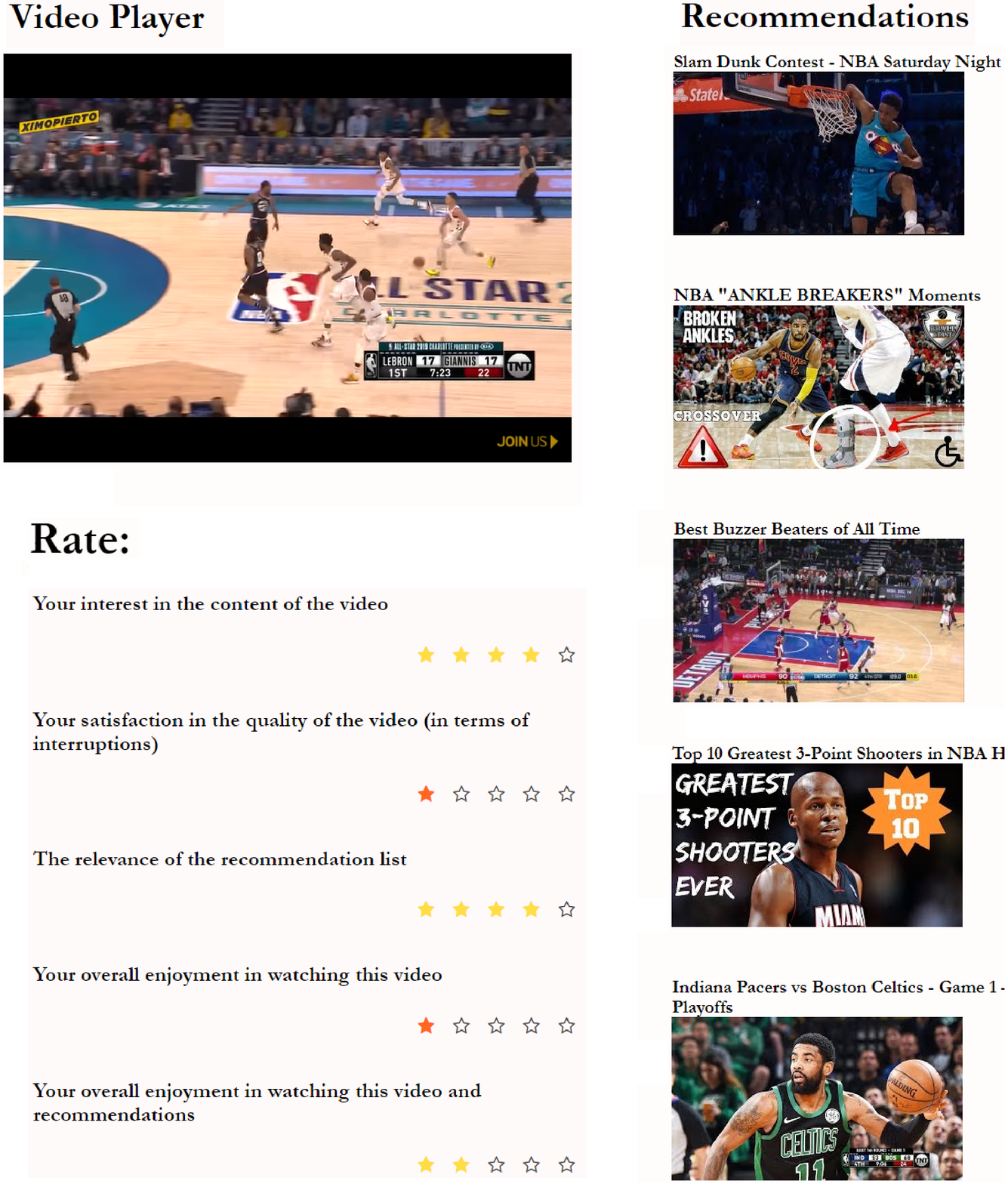}
    \caption{Experimental platform - instance of a user experiment: (i) a user watches a video (top/left), and is requested to (ii) rate her satisfaction from the watched video and recommendations (bottom/left) and (iii) select one of the recommendations to proceed to the following video (right).}
    \label{fig:exp-screenshot}
\end{figure}

\subsection{Experiment Setup}

\myitem{Region (Action 1).} We offer as options a subset of the regions provided by the YouTube API~\cite{youtube-api}; we selected $7$ representative regions (different continents, diverse demographics, available video data). 

\myitem{Initial list of videos (Action 2).} For each region, we retrieve from the YouTube API the list of $50$ top trending videos. We randomly select $20$ of them (for the selected region) to present to the user.

\myitem{Caching.} We compiled a list of $500$ videos IDs that are assumed to be cached\footnote{Note that we do not cache any video, since this is not allowed by the terms of use of the YouTube service.}
; we consider a different list per region. In each list, we select to first include the top $50$ trending videos in this region. Then, for each of these $50$ videos, we request its $50$ recommendations / related videos provided by YouTube API. From these $2500$ ($50\times50$) total videos, we add in the list the $450$ videos with the higher number of views (``most popular''). 

\myitem{List of recommendations (Action 3b).} The list of the $5$ recommendations given to the user when watching a video are generated by \ourAlgo. We tuned the parameters of \ourAlgo as follows: the width of the BFS is $50$ in the first depth, and for the first $10$ of the item found in the first depth we search in second depth as well and retrieve a list of $50$; in total we compile a list of $50+10\cdot 50 = 550$ videos. This modification compared to the parameters used in~\secref{sec:measurements} was done for scalability reasons (number of available credits, time needed by the YouTube API to respond, etc.). 

\myitem{Collected data (Action 3c).}
In each experiment session we collect the following data:
\begin{itemize}[leftmargin=*]
    \item ID of watched video
    \item IDs of the final recommendation list (i.e., the $5$ videos presented in the right side in Fig.~\ref{fig:exp-screenshot}), and the positions of videos in this list
    \item ID of the initial YouTube recommendations
    ; these videos were not presented to the user
    \item IDs of videos that are (assumed to be) cached
    \item User ratings
\end{itemize}

%% file: Experiments.tex
\myitem{Key finding: The CHR in the real-user experiments is 47\%.}

In our experiments with \ourAlgo recommendations, a percentage of 47\% among the 
videos selected and watched by real users, was for cached videos\footnote{Note that, similarly to the calculation of \eq{eq:chr-single-request}, this percentage does not include the first video views of the experiments (i.e., \textit{Action 2}), since all initial recommendations are for cached videos.}. While our experimental results are admittedly preliminary for a quantitative analysis, they qualitatively verify that we achieve in practice (i.e., with real users) the CHR values demonstrated in~\secref{sec:measurements}.

Moreover, in Fig.~\ref{fig:avgCHRstep} we present how the CHR (calculated as in \eq{eq:chr-sequential-request}) varies with the number of the sequence requests, i.e., when we move farther from the initial recommendations for the top popular (and cached) videos. We observe that our findings validate the corresponding measurement results in Fig.~\ref{fig:chr-vs-K-sequential}, i.e., as expected the CHR decreases (from around 70\% in the second step to 50\% after five steps), however, this decrease is not large.

\begin{figure}
\centering
\begin{minipage}[t]{0.46\linewidth}
\centering
\includegraphics[width=1\columnwidth]{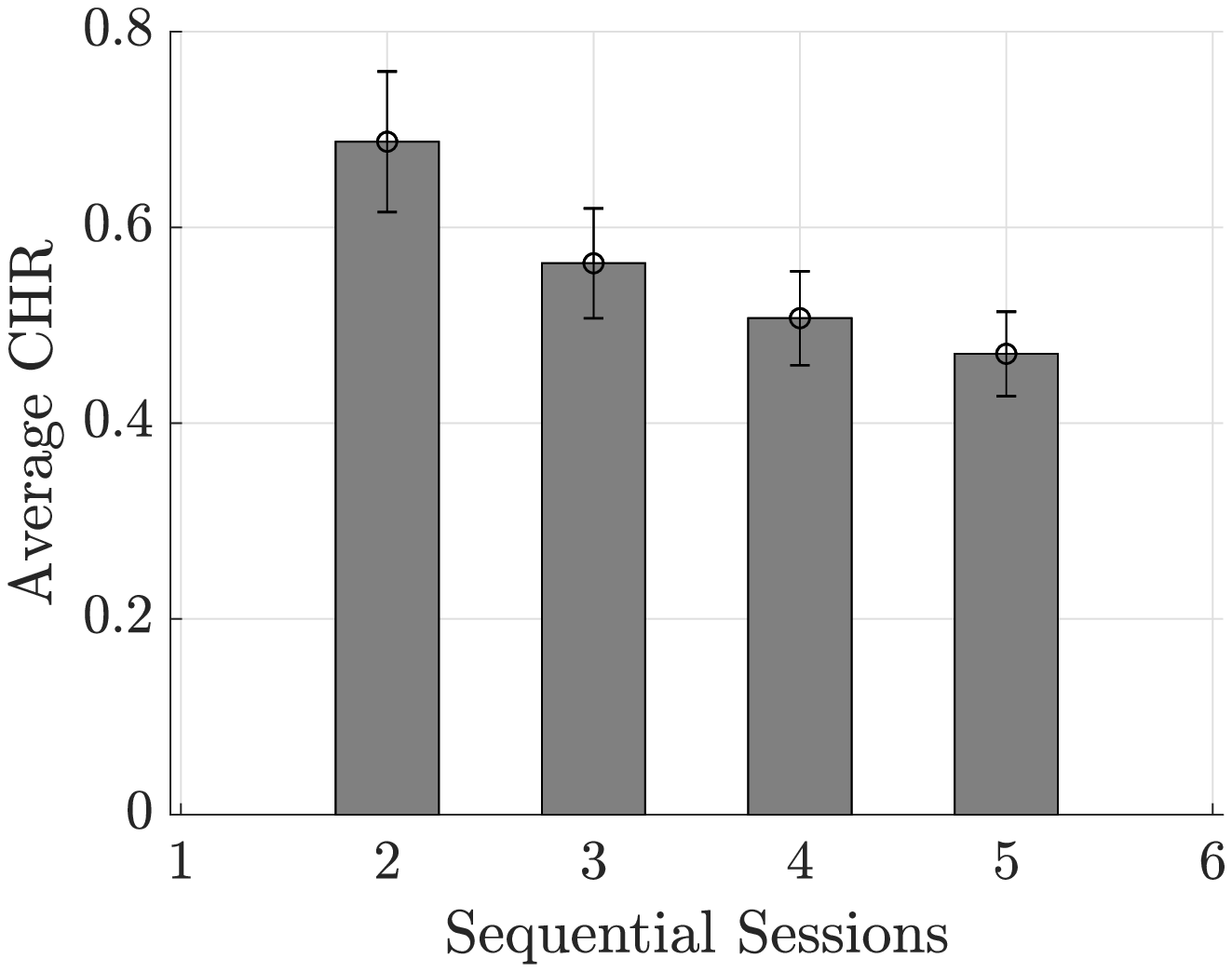}
\caption{CHR vs. \#requests in sequence $K$ ($C$=500, $W_{BFS}$=20, $D_{BFS}$=2).}
\label{fig:avgCHRstep}
\end{minipage}
\hspace{0.03\linewidth}
\begin{minipage}[t]{0.46\linewidth}
\centering
\includegraphics[width=1\columnwidth]{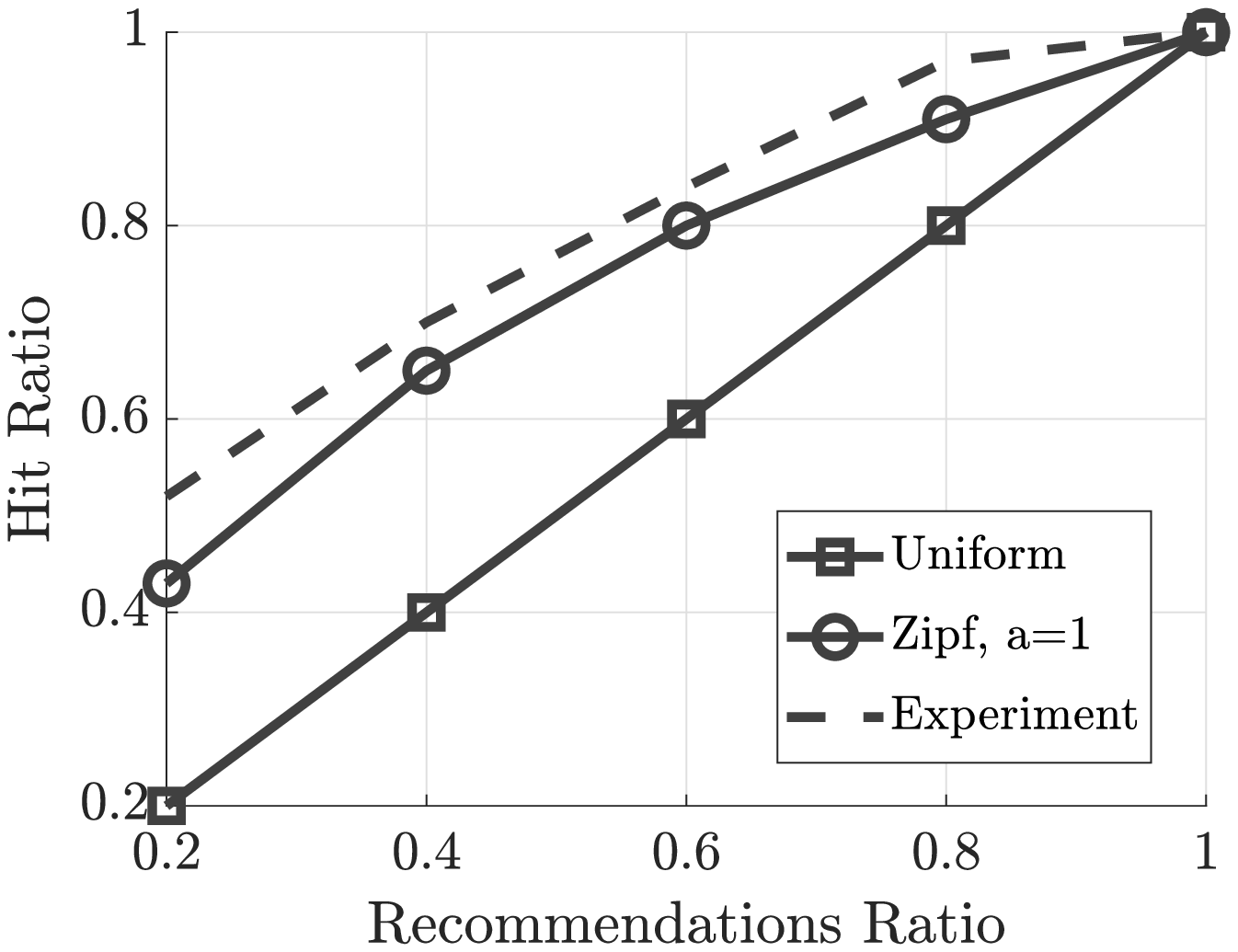}
\caption{CHR vs fraction of cached videos in recommendation list}

\label{fig:cd-vs-5}
\end{minipage}
\end{figure}

The observed decrease in the CHR, when moving from the top popular list is due to the fact that there can be found less (directly or indirectly) related contents that are cached. Table~\ref{tab:percentage-ccr0} shows the fraction of sessions, in which no cached content was found by \ourAlgo at $x^{th}$ request in sequence by a user. After five requests, in 11\% of the cases \ourAlgo did not find any cached related video to recommend (i.e., resulting in at least 11\% cache misses), while among the first requests this percentage is only 2\%.

\begin{table}[h]
\centering
\caption{Percentage of experiment samples in which none of the videos in the recommendation list was cached.}
\label{tab:percentage-ccr0}
\begin{tabular}{c|ccccc}
{Request step}     &  1&2&3&4&5\\
\hline
{\% experiment samples} &{2\%}&{5\%}&{8\%}&{10\%}&{11\%}
\end{tabular}
\end{table}


\begin{figure}

\centering
\begin{minipage}[t]{0.46\linewidth}
\centering
\includegraphics[width=1\columnwidth]{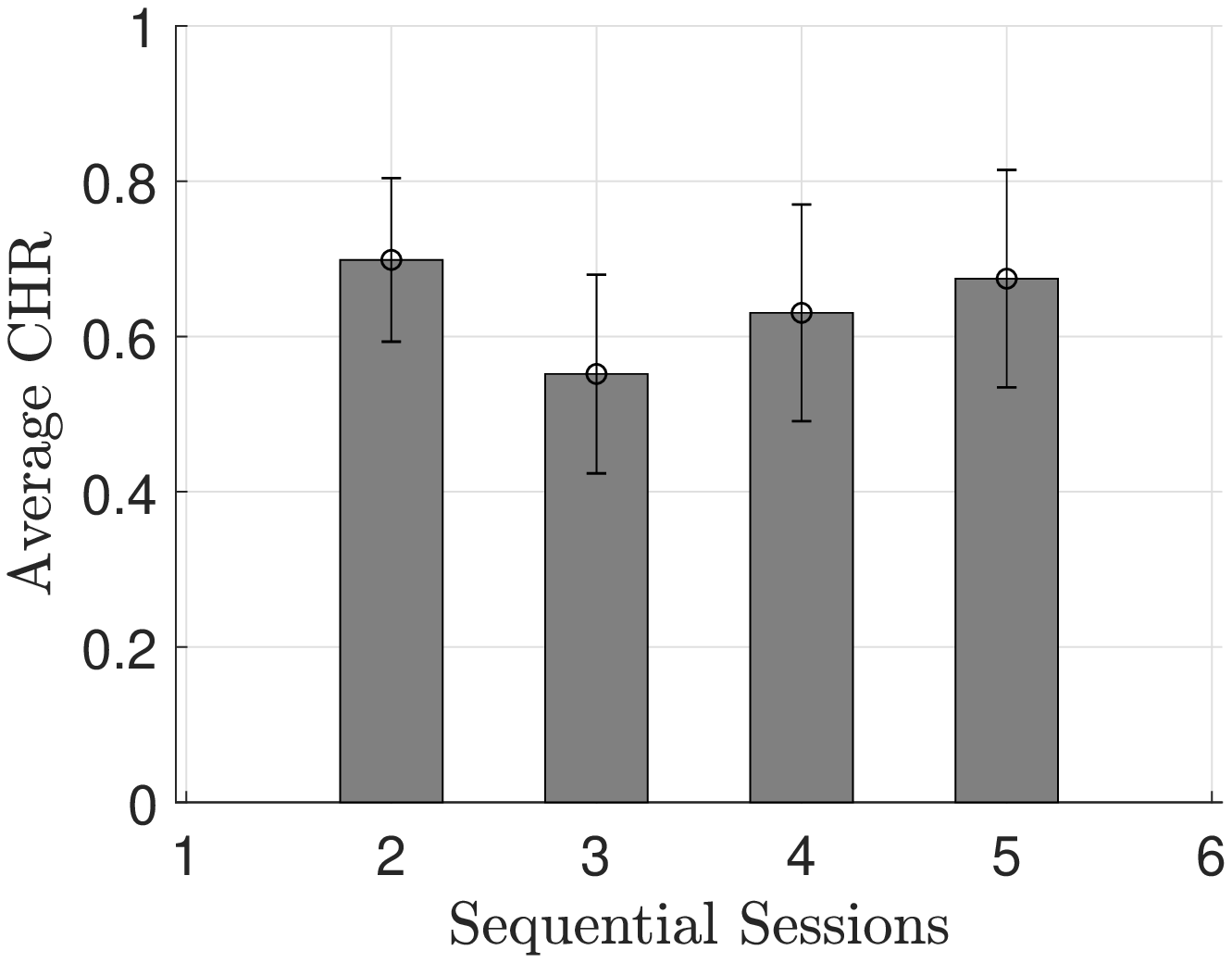}
\caption{CHR when at least one recommended content is cached vs. \#requests in sequence $K$ ($C$=500, $W_{BFS}$=20, $D_{BFS}$=2).}
\label{fig:crr1}
\end{minipage}
\hspace{0.03\linewidth}
\begin{minipage}[t]{0.46\linewidth}
\centering
\includegraphics[width=1\columnwidth]{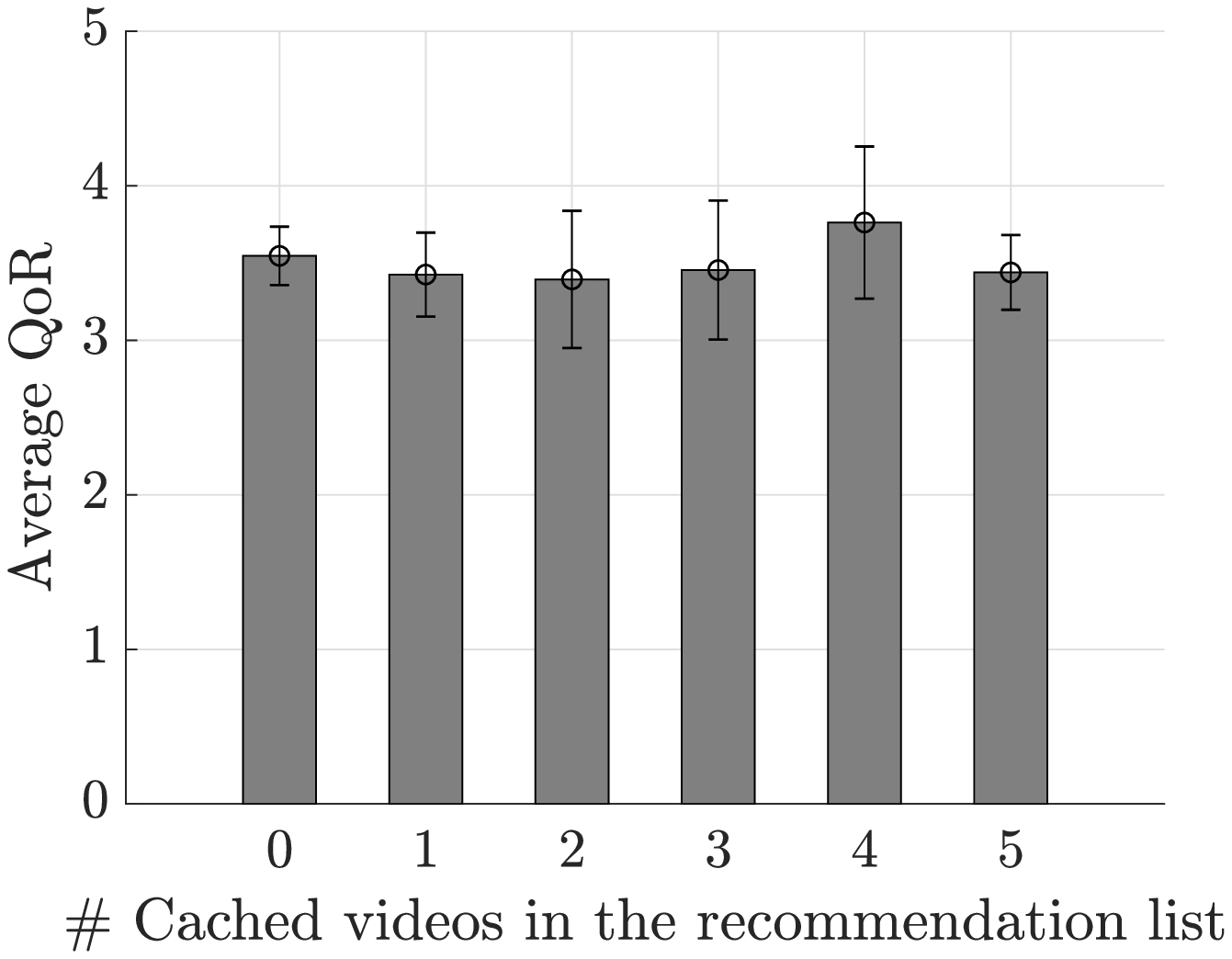}
\caption{QoR vs. \#cached videos in recommendation list}
\label{fig:avgqorVScrr}
\end{minipage}

\end{figure}

\mjr{Finally, we present in Table~\ref{tab:experiments-CHR-vs-C} the achieved CHR in our experiments by considering different number (and sets) of cached contents (i.e., fractions of the total 500 contents assumed cached in the baseline scenario)\footnote{While for this analysis we assume a fraction of the 500 cached contents, in the experiments \ourAlgo considered in its recommendations the initial set of the 500 contents. This means that the presented results may underestimate the best CHR that could achieved by \ourAlgo}. The increase in CHR is almost linear with $C$ in both scenarios, which is in line with the measurement results (Fig.~\ref{fig:chr-vs-C-d2}). Most popular caching is more efficient than random caching (as expected), and this effect of the caching policy becomes more important for smaller cache sizes.}

\begin{table}[h]
\centering
\caption{\mjr{CHR vs. cache size $C$, with the most popular (top row) or random (bottom row) contents being cached.}}
\label{tab:experiments-CHR-vs-C}
\begin{tabular}{l|cccccc}
{$C$}         & 50 & 100 & 200 & 300 & 400 & 500\\
\hline
{CHR (most popular)}&  0.11& 0.16& 0.24& 0.32& 0.40& 0.47\\
{CHR (random)}      & 0.05& 0.12& 0.19& 0.27& 0.38& 0.47
\end{tabular}
\end{table}

\myitem{Key finding: Users tend to select the top recommendations, even when those are ``nudged'' towards cached contents.}

However, what happens when at least one cached content can be recommended (i.e., is directly or indirectly related to the currently watched content)? Figure~\ref{fig:crr1} shows the CHR per step (x-axis) when at least one cached content is recommended by \ourAlgo (we remind that cached contents are placed in the top of the list). We can observe that the CHR is always more than 55\% (and up to 70\%), which strongly indicates that users select the \ourAlgo recommendations, when they are provided.

Figure~\ref{fig:cd-vs-5} shows in more detail the preference of videos with respect to the \ourAlgo recommendations. Specifically, the y-axis is the $CHR(x)$, i.e., the CHR in sessions where the \ourAlgo recommendations included $x$ cached videos (x-axis). The continuous line shows the $CHR(x)$ observed in our experiments, while the dashed lines correspond to an hypothesized \textit{uniform} selection of contents (i.e., the user selects randonmly one of the 5 presented recommendations), and \textit{Zipf} selection of contents. The main observation is that users tend to prefer the \ourAlgo recommendations presented in the top of the list; this behavior that has been previously reported for the YouTube service, does not seem to be affected by the fact that the recommendations are nudged towards cached videos. This indicates that \textit{using a carefully designed QoS-aware RSs does not have a negative impact on user preferences}.

Finally, these results also provide and insights on the tuning of \ourAlgo: The decrease in the CHR is mainly due to sessions where \ourAlgo did not find any cached video in the related list $\mathcal{L}$, and not due to the number of cached videos in the recommendation list. Hence, for these sessions (with 0 cached recommendations) we could tune \ourAlgo to search in larger depth/width for cached contents; even finding one such content and placing it in the top of the recommendation list, would lead to increased CHR.

\myitem{Key finding: The recommendations generated by \ourAlgo are perceived as high-quality by users.}

The results presented above, demonstrate that applying an algorithm like \ourAlgo in practice, could indeed lead to performance gains, since users are willing to select the nudged recommendations towards cached videos. 
Apart from the network benefits, in this last part of our analysis, we investigate whether the \ourAlgo recommendations satisfy the user: Do the users select the nudged recommendations because the find them appealing or because they do not have a (much) better alternative? Are they satisfied by the recommended videos?

In the experiments, we asked the users to provide ratings for the relevance of the recommendation list (QoR) and their interest in the watched video. Figure~\ref{fig:avgqorVScrr} shows the average rating for QoR (y-axis) in sessions where $x$ out of the $5$ recommendations are for cached videos. We observe that the users do not significantly differentiate, in terms of QoR, between the initial YouTube recommendations ($x=0$) from the \ourAlgo recommendations ($x>0$). This clearly shows that the nudged recommendations are not perceived as intrusive by the user.


In addition to this, we investigate whether the users ultimately liked the video they selected to watch (and, e.g., were not misled by the recommendation). Table~\ref{tab:interest-ratings-pmf} shows the distribution of the \textit{Interest} ratings for the cached and non-cached videos. The interest for contents from initial YouTube recommendations (i.e., all non-cached videos) is not significantly different than the interest in the cached videos that the users watched. This further supports our arguments and provide experimental evidence that (a) \ourAlgo can find high-quality recommendations, and (b) nudging recommendations towards cached video, does not have a significant negative impact in user interest.

\begin{table}[h]
\centering
\caption{Percentage of responses per \textit{Interest} rating for the cached and non-cached videos.}
\label{tab:interest-ratings-pmf}

\begin{tabular}{l|ccccc}
{} & \multicolumn{3}{c}{{Rating of \textit{Interest}}}\\
{}
    & {1-2}$\bigstar$ &{3}$\bigstar$ &{4-5}$\bigstar$ \\
\hline 
{Non-cached videos} &  24\% &   19\% & 57\% \\
{Cached videos} & 25\% &   18\% &   57\%
\end{tabular}
\end{table}

%% file: Related.tex
\blue{The joint network and recommendations paradigm has been recently introduced, in the context of soft cache hits~\cite{sch-chants-2016,sermpezis2018soft} or network-friendly recommendations~\cite{chatzieleftheriou2017caching,chatzieleftheriou2019jointly}, aiming to jointly design the content caching policy and the recommendation policy in order to achieve higher cache hit rates. The promising gains in the caching efficiency (which comes ``for free'' from a technology point of view, e.g., without extra investment in equipment or new communication technologies) demonstrated by these early works, motivated more work on the topic~\cite{giannakas-wowmom-2018,
zhu2018coded,lin2018joint,song2018making,qi2018optimizing,giannakas2019order,chatzieleftheriou2019joint,gupta2019effect,lin2019content,assis2019recomendaccao,sermpezis2019towards,
costantini2019approximation,garetto2020similarity,tsigkari2020user,li2020leveraging}.}

\blue{The majority of related works considers \textit{cache-aware} recommendations in mobile edge caching~\cite{sch-chants-2016,chatzieleftheriou2017caching,sermpezis2018soft,giannakas-wowmom-2018,zhu2018coded,chatzieleftheriou2019jointly,costantini2019approximation,garetto2020similarity,tsigkari2020user,li2020leveraging} to improve the cache hit ratio by optimizing the recommendation and/or caching policies. However, the same principles can easily generalize to \textit{network-aware} recommendations, where each content can be delivered by the network with a given cost or quality~\cite{giannakas2019order}. Other aspects considered in literature include coded caching~\cite{zhu2018coded}, broadcast communications with coded transmissions~\cite{lin2018joint,song2018making,lin2019content}, the extra dimension of user association to small base stations~\cite{chatzieleftheriou2019joint}, or swarming systems~\cite{content-recommendation-swarming}.} \mjr{A similar concept is similarity caching~\cite{garetto2020similarity}, with can have more generic applications (e.g., machine learning tasks) than multimedia services.}

\blue{Our work is complementary to previous works that study techniques for optimizing the network performance. To our best knowledge, all existing studies have evaluated the performance in simulation setups. On the contrary, we focus on realistic evaluations of the joint network and recommendations. Our goal was to (i) enable researchers perform realistic evaluations of their solutions, (ii) verify that the claimed performance gains can hold also in practice (i.e., in real setups), (iii) provide evidence for the assumptions made by previous works that users will be willing to follow ``nudged'' network-aware recommendations.}

\mjr{Finally, while in this paper we focused on the YouTube case, the proposed methodology is generalizable to other services and settings. The simplicity of the \ourAlgo algorithm makes it easily implementable, without this having a negative effect on performance, as shown by our results or, e.g., the evaluation in~\cite{li2020leveraging} for short-video services scenarios where \ourAlgo achieves comparable performance to state-of-the-art schemes~\cite{li2020leveraging,zhang2019challenges_short_video}.}

%% file: Conclusions.tex
In this paper, we proposed a methodology that enables to evaluate joint network and recommendation techniques in realistic setups, by leveraging available information from real recommendation systems. Enabling realistic evaluations of previous or future works can be important for fine-tuning the parameters and assumptions of the proposed solutions, as well as provide insights for potential practical challenges. 

Our results on the YouTube video service showed that the significant gains that have been indicated in related literature, can be achieved in practice as well. This is a positive message for the feasibility and benefits of the joint network and recommendations paradigm. To further strengthen this message, we conducted experiments with real users to investigate the feasibility from the user perspective; our findings are the first to provide experimental evidence that network-aware recommendations can be perceived as non-intrusive by users (a major assumption in related work).

\mjr{We believe that our findings can motivate further research on the topic. For instance, large-scale experiments with users or measurements in real network conditions could provide useful results and insights for the design of operational network-aware recommendation systems.}

%% file: biographies.tex
\begin{IEEEbiographynophoto}
{Savvas Kastanakis}
received  the  B.Sc. degree  in  Computer  Science in 2018 and currently pursues an M.Sc. degree in Telecommunications and Networks at the Computer Science Department, University of Crete. His research interests include: Internet Measurements, Internet of Things, Network Security. He won second place in the 1st ACM SIGCOMM Hackathon in 2018. He is currently an ACM Student Member.
\end{IEEEbiographynophoto}

\begin{IEEEbiographynophoto}
{Pavlos Sermpezis}
received the Diploma in Electrical and Computer Engineering from the Aristotle University of Thessaloniki (AUTH), Greece, and a PhD in Computer Science and Networks from EURECOM, Sophia Antipolis, France. He was a post-doctoral researcher at FORTH, Greece, and currently is a post-doctoral research at the Computer Science Dept. at AUTH, Greece. His main research interests are in modeling and performance analysis for communication networks, network measurements, and data science.
\end{IEEEbiographynophoto}

\begin{IEEEbiographynophoto}
{Vasileios Kotronis}
received a Diploma in Electrical and Computer Engineering from the National Technical University of Athens, Greece and a PhD in Information Technology and Electrical Engineering from ETH Zurich, Switzerland. He is currently a post-doctoral researcher at FORTH, Greece. His main research interests include: Internet routing and measurements, software defined networking, and network security.
\end{IEEEbiographynophoto}

\begin{IEEEbiographynophoto}
{Daniel Sadoc Menasch\'e} received the Ph.D.degree in computer science from the University of Massachusetts, Amherst, in 2011. He is currently an Assistant Professor with the Computer Science Department, Federal University of Rio de Janeiro, Brazil. His research interests are in modeling, analysis, security, and performance evaluation of computer systems. He was a recipient of the best paper awards at GLOBECOM 2007, CoNEXT 2009, INFOCOM 2013, and ICGSE2015. He is currently an Affiliated Member of the Brazilian Academy of Sciences.
\end{IEEEbiographynophoto}

\begin{IEEEbiographynophoto}
{Thrasyvoulos Spyropoulos}
received the Diploma in Electrical and Computer Engineering from the National Technical University of Athens, Greece, and a Ph.D degree in Electrical Engineering from the University of Southern California. He was a post-doctoral researcher at INRIA and then, a senior researcher with the Swiss Federal Institute of Technology (ETH) Zurich. He is currently an Assistant Professor at EURECOM, Sophia-Antipolis. He is the recipient of the best paper award in IEEE SECON 2008, and IEEE WoWMoM 2012.
\end{IEEEbiographynophoto}

%% file: JointCacheRec.tex
\section{Caching Optimization under \ourAlgo}\label{appendix:joint-cache-rec}

\subsection{Problem Formulation and Optimization Algorithm}
In the following, we first analytically formulate and study the problem of optimizing the caching policy under \ourAlgo recommendations, and propose an approximation algorithm with provable performance guarantees. 

Let a content catalog $\mathcal{V}$, $V=|\mathcal{V}|$, and a content popularity vector $\mathbf{q} = [q_{1}, ..., q_{V}]^{T}$. Let $\mathcal{L}(v)\subseteq \mathcal{V}$ be the set of contents that are explored by \ourAlgo (at \textit{line 1}) for a content $v\in\mathcal{V}$.

For some set of cached contents $\mathcal{C}\subseteq\mathcal{V}$, and a content $v$, \ourAlgo returns a list of recommendations $\mathcal{R}(v)$ ($|\mathcal{R}(v)|=N$)
. Therefore, CHR can be expressed as
\begin{equation}\label{eq:chr-generic}
CHR(\mathcal{C}) = \sum_{v\in\mathcal{V}}q_{v}\sum_{i=1}^{N(\mathcal{C},v)} p_{i}
\end{equation}
where $N(\mathcal{C},v) = \min\{|\mathcal{C}\cap\mathcal{L}(v)|, N\}$, and $p_{i}$ is the probability for a user to select the $i^{th}$ recommended content.

Then, the problem of optimizing the caching policy (to be jointly used with \ourAlgo), is formulated as follows:
\begin{equation}\label{optim-problem}
\max_{\mathcal{C}} ~CHR(\mathcal{C})~~~~~\textrm{s.t.,} |\mathcal{C}|\leq C
\end{equation}
where $C$ is the capacity of the --MEC-- cache.
We prove the following for the optimization problem of \eq{optim-problem}.
\begin{mylemma}
The optimization problem of~\eq{optim-problem}: (i) is NP-hard, (ii) cannot be approximated within $1-\frac{1}{e}+o(1)$ in polynomial time, and (iii) has a monotone (non-decreasing) submodular objective function, and is subject to a cardinality constraint.
\end{mylemma}
\begin{proof}
Items (i) and (ii) of the above lemma, are proven by reduction to the \textit{maximum set coverage} problem, and we prove item (iii) using standard methods (see, e.g., similar proofs in~\cite{femto,sermpezis-sch-globecom}). 
\end{proof}

If we design a greedy algorithm that starts from an empty set of cached contents $\mathcal{C}_{g}=\emptyset$, and at each iteration it augments the set $\mathcal{C}_{g}$ (until $|\mathcal{C}_{g}|= C$) as follows:
\begin{equation}\label{eq:greedy-algo}
\mathcal{C}_{g}\leftarrow \mathcal{C}_{g}\cup \arg\max_{v\in\mathcal{V}} CHR(\mathcal{C}_{g}\cup\{v\}),
\end{equation}
then the properties stated in item (iii) satisfy that it holds~\cite{krause2012submodular}
\begin{equation}\label{eq:greedy-bound}
CHR(\mathcal{C}_{g}) \geq \left(1-\frac{1}{e}\right)\cdot CHR(\mathcal{C}^{*})
\end{equation}
where $\mathcal{C}^{*}$ the optimal solution of the problem of~\eq{optim-problem}. 

\textit{Remark}: While \eq{eq:greedy-bound} gives a lower bound for the performance of the greedy algorithm, in practice greedy algorithms have been shown to perform often very close to the optimal~\cite{bian2017guarantees}.

%% file: measurements_greedy.tex
\blue{We investigate the performance when the list of cached contents is selected to optimize the CHR by using the greedy algorithm
. We consider both ``front-page'' and ``search bar'' video demands.}

\myitem{Efficiency vs. scalability.} Calculating the CHR from \eq{eq:chr-generic} requires running a BFS (\ourAlgo, \textit{line 1}) and generating the lists $\mathcal{L}(v)$, for every content $v\in \mathcal{V}$. In practice, for scalability reasons, the most popular contents (i.e., with high $q_{i}$) can be considered by the greedy algorithm in the calculation of the objective function \eq{eq:chr-generic}, since those contribute more to the objective function. \blue{To demonstrate the involved trade-offs between scalability and performance, we consider two scenarios with synthetic content catalogs of size $|\mathcal{V}|$=1000 and $|\mathcal{V}|$=10000 (where content popularity $q_{i}$ follows a Zipf(a=1) distribution, and each content is related on average with $10$ other contents), and calculate the CHR achieved by \ourAlgo (N=10, $D_{BFS}$=2, $W_{BFS}$=5) when the greedy algorithm considers only a fraction $\mathcal{V}^{'}$ of the entire catalog, $\mathcal{V}^{'}\subseteq \mathcal{V}$, and a cache of size $C$=10. Table~\ref{tab:cabaret-greedy-vs-most-popular} presents the achieved CHR, normalized over the maximum CHR achieved when considering the entire catalog $\mathcal{V}$. We can see that even considering very small fractions of the content catalog in the caching decisions, can still achieve significant performance, while considering a 10\% of the content catalog can already achieve 90\% and 86\% of the maximum performance in the case of $|\mathcal{V}|$=1000 and $|\mathcal{V}|$=10000, respectively. This indicates that \ourAlgo-like approaches can be an efficient and scalable in real systems with very large content catalogs.}


\begin{table}[h]
\centering
\caption{\blue{CHR under caching with the greedy algorithm considering only a fraction of the most popular contents of the catalog, $i\in \mathcal{V}^{'}\subset \mathcal{V}$; values are normalized over the maximum achievable performance.}}
\label{tab:cabaret-greedy-vs-most-popular}
\begin{tabular}{cl|cccc}
\multicolumn{2}{c|}{fraction of the catalog $\frac{|\mathcal{V}^{'}|}{|\mathcal{V}|}$ }
    & 0.1\% & 1\% & 5\% & 10\%  \\
\hline
\multirow{2}{*}{$\frac{CHR(\mathcal{V}^{'})}{CHR(\mathcal{V})}$ }
& $|\mathcal{V}|=1000$
    & 0.72 & 0.74 & 0.89 & 0.90\\
& $|\mathcal{V}|=10000$            
    &  0.40 & 0.54  & 0.82  & 0.86 
\end{tabular}
\end{table}

\blue{The reason that the greedy algorithm remains efficient even with this simplification, is that }%
any video in the catalog is still candidate to be cached, e.g., a video with low $q_{i}$ can bring a large increase in the CHR through its association with many popular contents. In fact, in our experiments, for the calculation of \eq{eq:chr-generic}, we consider only the 50 most popular videos, for which we set $q_{i} = \frac{1}{50}$. Nevertheless, in the different scenarios we tested, only 10\% to 30\% of the cached videos (selected by the greedy algorithm) were also in the top 50 most popular.

\myitem{Greedy vs. Top caching.} In Fig.~\ref{fig:chr-vs-c-greedy}, we compare the achieved CHR for ``Front Page'' video demand, when the cache is populated according to the greedy algorithm of \eq{eq:greedy-algo} (\textit{Greedy Caching}) and with the top most popular videos (\textit{Top Caching}). \textit{Greedy caching} always outperforms \textit{top caching}, with an increase in the CHR of around a factor of 2 for uniform video selection (for the Zipf($a$=1) scenarios we tested, the CHR values are even higher, and the relative performance is 1.5 times higher). This clearly demonstrates that the gains from joint recommendation and caching~\cite{sermpezis-sch-globecom,chatzieleftheriou2017caching}, are applicable even in simple practical scenarios (e.g., \ourAlgo \& greedy caching). Finally, while \textit{greedy caching} increases the CHR even with regular YouTube recommendations, the CHR is still less than $50\%$ of the \ourAlgo case with \textit{top caching}. This further stresses the benefits from \ourAlgo's cache-aware recommendations.

\begin{figure}\centering
\centering
\subfigure[ ``Front Page'' video demand]{
\includegraphics[width=0.47\linewidth]{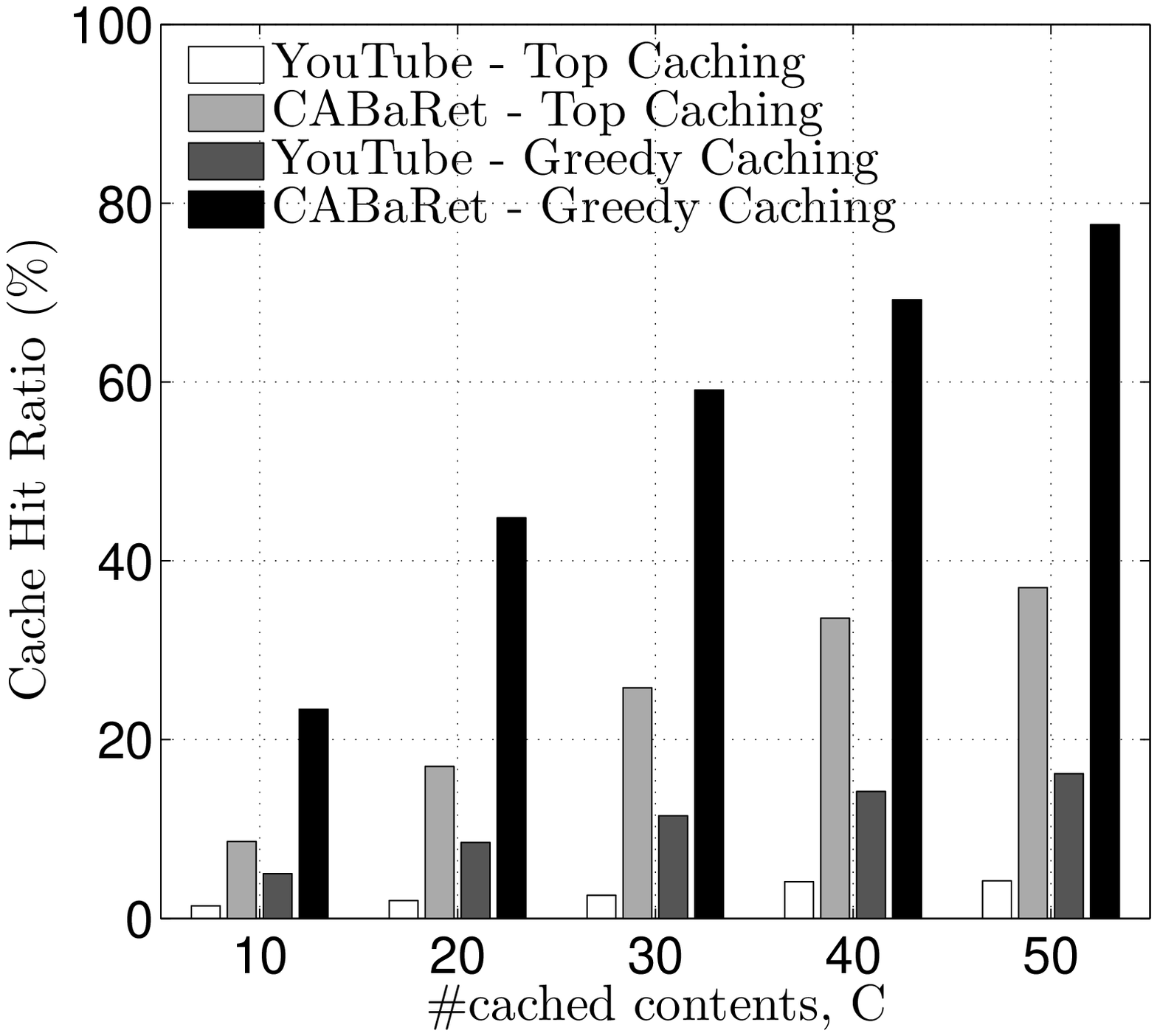}\label{fig:chr-vs-c-greedy}}
\subfigure[``Search Bar'' video demand]{
\includegraphics[width=0.47\linewidth]{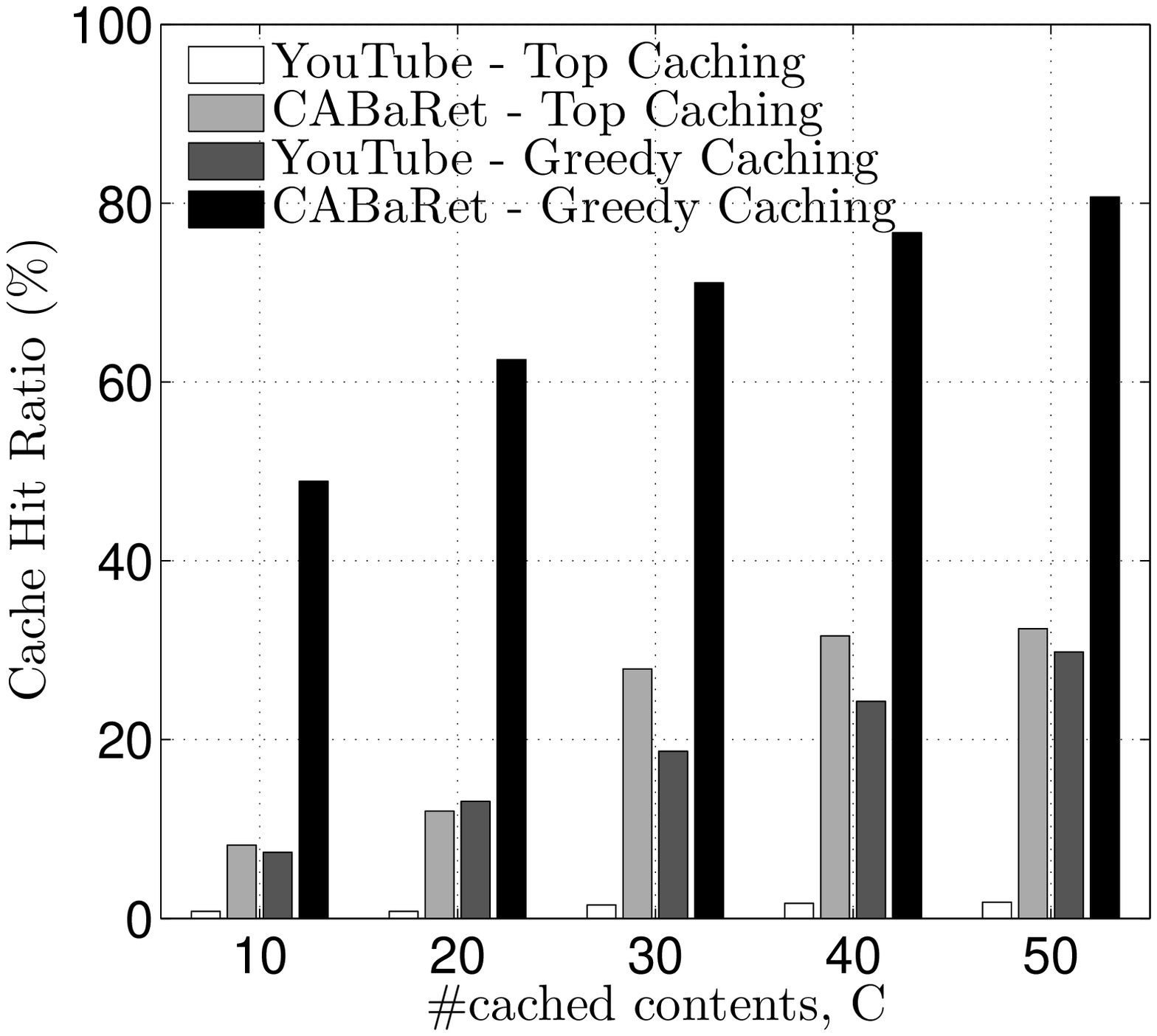}\label{fig:greedyCachingZipf1GR}}
\caption{CHR vs. \# cached contents with \ourAlgo parameters $W_{BFS}$=50 and $D_{BFS}$=2, and video demand (a) ``Front Page'' video demand with $p_{i}\sim$\textit{uniform}, and (b) ``Search Bar'' video demand with  $p_{i}\sim$\textit{Zipf(a=1)} 
.}
\end{figure}

\blue{Similar findings can be seen in Fig.~\ref{fig:greedyCachingZipf1GR} for scenarios with ``Search Bar'' video demand. A difference is that in these scenarios the CHR under YouTube recommendations with \textit{greedy caching} is comparable to \ourAlgo recommendations with \textit{top caching}, which indicates that similar performance can be achieved by carefully selecting either only the recommendations (\ourAlgo+ \textit{top caching}) or only the caching (YouTube + \textit{greedy caching}). However, when combining both (\ourAlgo + \textit{greedy caching}), increases more than two times the CHR.}